\documentclass[a4paper,11pt]{article}
\usepackage[margin=0.9in]{geometry}

\usepackage[english]{babel}
\usepackage{amsfonts}
\usepackage{amssymb}
\usepackage{amsthm, amsmath}
\usepackage{eucal}
\usepackage{mathrsfs}
\allowdisplaybreaks[0]
\usepackage{hyperref}
\usepackage{bm}
\usepackage{graphicx}
\usepackage{caption}
\usepackage{subcaption}
\usepackage{algorithm,algorithmic}
\usepackage{framed}
\usepackage{diagbox}
\usepackage[titletoc,title]{appendix}
\usepackage{colortbl}
\usepackage{pifont}
\usepackage{xcolor}
\usepackage{breakcites}
	
\usepackage{lineno}
\usepackage{bm}

\usepackage{multirow}

\usepackage{comment}

\definecolor{Gray}{gray}{0.9}

\newcommand{\xmark}{\ding{55}}

\usepackage{todonotes}

\newcommand{\R}{\mathbb{R}}

\newcommand{\be}{\begin{equation}}
\newcommand{\ee}{\end{equation}}

\DeclareMathOperator{\bR}{\mathbb{R}}

\newtheorem*{theorem}{Theorem}

\numberwithin{equation}{section}

\begin{document}

\title{Visual illusions via neural dynamics: Wilson-Cowan-type models and the efficient representation principle}

\author{Marcelo Bertalm\'io\\
DTIC, Universitat Pompeu Fabra, Barcelona, Spain\\
\href{mailto:marcelo.bertalmio@upf.edu}{marcelo.bertalmio@upf.edu}
\and
Luca Calatroni\\
Universit\'e C\^ote d'Azur, CNRS, INRIA,\\
Laboratoire d'Informatique, Signaux et Syst\`emes de Sophia Antipolis,  France\\
\href{mailto:calatroni@i3s.unice.fr}{calatroni@i3s.unice.fr}
\and
Valentina Franceschi\\
IMO, Universit\'e Paris-Sud, Orsay, France\\
\href{mailto:valentina.franceschi@u-psud.fr}{valentina.franceschi@u-psud.fr}
\and
Benedetta Franceschiello\\
 FAA, LINE, Radiology, CHUV, Lausanne, Switzerland\\
\href{mailto:benedetta.franceschiello@fa2.ch}{benedetta.franceschiello@fa2.ch}
\and
Alexander Gomez-Villa\\
DTIC, Universitat Pompeu Fabra, Barcelona, Spain\\
\href{mailto:alexander.gomez@upf.edu}{alexander.gomez@upf.edu}
\and 
 Dario Prandi \\
Universit\'e Paris-Saclay, CNRS, CentraleSupélec, \\ Laboratoire des signaux et syst\`emes, Gif-sur-Yvette, France\\
\href{mailto:dario.prandi@l2s.centralesupelec.fr}{dario.prandi@l2s.centralesupelec.fr}
 }
 
 \date{}

\maketitle

\begin{abstract}
% REVIEWER COMMENT - The abstract is a bit repetitive and not very clear. The term "efficient representation principle" is used twice, and a sentence such as "... the neural dynamics equations consider explicitly the orientation ..." is not clear without reading the Ms.
% ORIGINAL VERSION:
%In this work we have aimed to reproduce supra-threshold perception phenomena, specifically visual illusions, with Wilson-Cowan-type models of neuronal dynamics. We have found that it is indeed possible to do so, but that the ability to replicate visual illusions is related to how well the neural activity equations comply with the efficient representation principle. Our first contribution is to show that the Wilson-Cowan equations can reproduce a number of brightness and orientation-dependent illusions, and that the latter type of illusions require that the neuronal dynamics equations consider explicitly the orientation, as expected. Then, we formally prove that there can't be an energy functional that the Wilson-Cowan equations are minimizing, but that a slight modification makes them variational and yields a model that is consistent with the efficient representation principle. Finally, we show that this new model provides a better reproduction of visual illusions than the original Wilson-Cowan formulation.
% UPDATED NOV 22ND
We reproduce supra-threshold perception phenomena, specifically visual illusions, by Wilson-Cowan-type models of neuronal dynamics. Our findings show that the ability to replicate the illusions considered is related to how well the neural activity equations comply with the efficient representation principle. Our first contribution consists in showing that the Wilson-Cowan (WC) equations can reproduce a number of brightness and orientation-dependent illusions. Then, we formally prove that there can't be an energy functional that the Wilson-Cowan dynamics are minimizing. This leads us to consider an alternative, variational modelling which has been previously employed for local histogram equalization (LHE) tasks. In order to adapt our model to the architecture of V1, we perform an extension that has an explicit dependence on local image orientation.  Finally, we report several numerical experiments showing  that LHE provides a better reproduction of visual illusions than the original WC formulation and that its cortical extension is capable to reproduce also complex orientation-dependent illusions.
\end{abstract}

%\textbf{Keywords}: Wilson-Cowan equations; Brightness perception; Efficient representation principle; Variational modelling  \\

\textbf{New \& Noteworthy}: We show that the Wilson-Cowan equations can reproduce a number of brightness and orientation-dependent illusions. Then, we formally prove that there can't be an energy functional that the Wilson-Cowan equations are minimizing, making them sub-optimal with respect to the efficient representation principle. We thus propose a slight modification that is consistent with such principle and show that this provides a better reproduction of visual illusions than the original Wilson-Cowan formulation. We also consider the cortical extension of both models in order to deal with more complex orientation-dependent illusions.

\section{Introduction}

The goal of this work is to point out the intimate connections existing between three popular approaches in vision science: the Wilson-Cowan equations, the study of visual brightness illusions, and the efficient representation theory.

As other articles in this special issue make abundantly clear, Wilson-Cowan equations have a long and successful story of modelling cortical low-level dynamics \cite{Cowan2016}.  Nonetheless, the study of psychophysics by Wilson-Cowan equations (\cite{Adini1997,Herzog2003,BertalmioJPP2009,Ernst2016,BertalmioVSS2017,Wilson2003,Wilson2007,Wilson2017}) is a topic that hasn't been addressed much in neuroscience,
%and to the best of our knowledge there is no recent publication
%\textcolor{blue}{to our knowledge there are not publications} \textcolor{red}{
and we are not aware of publications
in which Wilson-Cowan equations are used for predicting brightness illusions. In this work, we aim to fill this gap.

The study of visual illusions has always been key in the vision science community,
as the mismatches between reality and perception provide insights that can be very useful to develop new models of visual perception \cite{Kingdom2011} or of neural activity \cite{Eagleman2001,Murray2013}, and also to validate the existing ones.
It is commonly accepted that visual illusions arise due to neurobiological constraints \cite{Purves2008} that modify the underpinned mechanisms
%\textcolor{blue}{modify the underpinned mechanisms} \textcolor{red}{limit the ability} 
of the visual system. 

The efficient representation principle, introduced by Attneave \cite{Attneave1954} and Barlow \cite{Barlow1961},
states that neural responses aim to overcome these neurobiological constraints and to optimize the limited biological resources by being tailored to the statistics of the images that the individual typically encounters, so that visual information can be encoded in the most efficient way.
This principle is a general strategy observed across mammalian, amphibian and insect species \cite{Smirnakis1997} and is embodied by neural processing according to abundant
experimental evidence \cite{Fairhall2001,Mante2005,Benucci2013}.
%These constraints impose that visual processes adapt the visual encoding to the statistics of natural scenes so as to provide an efficient representation of them \cite{Olshausen2000}.

%\textcolor{blue}{This work aims at pulling together the three theories described before, providing a more unified framework to undisclose vision and its mechanisms}
%\textcolor{red}{The contributions of this work are threefold.}
Our work aims at pulling together the three approaches just mentioned, providing a more unified framework to understand vision mechanisms.
First, we show that the Wilson-Cowan equations are able to qualitatively reproduce a number of visual illusions.
%, while they are unable to replicate the perceptual phenomena induced by the chosen representatives of brightness and orientation-dependent illusions.
Secondly, we formally prove that Wilson-Cowan equations (with constant input) are not variational, in the sense that they are not minimizing any energy functional.
Next, we detail how a simple modification turning the Wilson-Cowan equations variational yields a local histogram equalisation method that is consistent with the efficient representation principle. We finally show how this new formulation provides a better reproduction of visual illusions than the Wilson-Cowan model.

We remark that our model has to be intended as a proof of concept, whose objective is the reproduction of perceptual phenomena at a macroscopic level with no quantitative assessment on analogous psychophysical data. There are in fact very important limitations for doing that, since such comparison would require both a perfect knowledge of how behavioural data were collected, and a tuning of the model parameters to match with the observed perception. Nonetheless, we believe that the numerical evidence of our experiments and our theoretical considerations can be used for future research studies comparing our computational results with the ones corresponding to experiments coming from psychophysics.

\section{Materials and methods}

\subsection{Visual illusions}  \label{sec:visual_ill}

%Visual illusions have always been considered as a window between reality and perception, enabling neuroscientists to disentangle the complicated process of vision \cite{Eagleman2001, Murray2013}. 
Computational models able to reproduce visual illusions represent very effective methods to test new hypotheses and generate new insights, both for neuroscience and applied disciplines such as image processing.
Illusions can be classified according to the main feature detection mechanisms involved during the visual process \cite{shapiro2016oxford}. In this contribution we considered two main groups of visual illusions to assess the efficacy of our model in reconstructing the perceptual process: \textit{brightness illusions}
%, where image regions with the same gray level are perceived as having different brightness, 
and \textit{orientation-dependent illusions}.
%, where the perceptual phenomena (e.g. in terms of brightness or contrast) is affected by the orientation of the image elements.
%background gray-level induces a misperception of the gray intensity level of the central targets. For instance, the effect is clear in the illusions presented in Fig.~\ref{fig:summary_brightness_illusion}: patterns with the same gray level and opacity appear darker or brighter depending on the background gray level intensity. On another hand, \textit{orientation-and-brightness-dependent illusions} like those shown in Fig.~\ref{fig:summary_orientation_illusion}, are brightness illusions where the orientation of the image elements plays a key role in the perceived effect.

\subsubsection{Brightness illusions}
Brightness illusions are a class of phenomena where 
image regions with the same gray level are perceived as having different brightness, depending on the shapes, arrangement and gray level of the surrounding elements.
Fig. \ref{fig:summary_brightness_illusion} shows the nine brightness illusions we have chosen to perform tests on in this paper. They are all very popular and at the same time they represent a diverse set, as we can see from the following descriptions.

\paragraph{White's illusion:} the left gray rectangle appears darker than the right one, while both are identical \cite{white1979new} (Fig. \ref{fig:summary_brightness_illusion}(a)).
%two gray rectangles having the same colour and opacity appear as having different gray-level intensities due to the underlying background, see \cite{white1979new}. The left box seems darker than the right one. See Fig. \ref{fig:summary_brightness_illusion}(a).

\paragraph{Simultaneous brightness contrast:} the left gray square appears lighter than the right one, while both are identical \cite{bruke} (Fig. \ref{fig:summary_brightness_illusion}(b)).
%the different backgrounds induce a change in brightness perception of the gray central boxes, although they have the same colour and opacity, see \cite{bruke}. See Fig. \ref{fig:summary_brightness_illusion}, second on the top from left to right.

\paragraph{Checkerboard illusion:} the  mid-gray square in the fifth column appears darker than the one in the seventh column, while both are identical 
\cite{devalois1990spatial} (Fig. \ref{fig:summary_brightness_illusion}(c)).

\paragraph{Chevreul illusion:} a pattern of homogeneous bands of increasing intensity from left to right is presented. However, the bands in the image are perceived as inhomogeneous, i.e. darker and brighter lines appear at the borders between adjacent bands \cite{ratliff1965mach} (Fig. \ref{fig:summary_brightness_illusion}(d)).

\paragraph{Chevreul cancellation:} when the order of the bands is reversed, now decreasing in intensity from left to right, the effect is cancelled \cite{geier2011changing} (Fig. \ref{fig:summary_brightness_illusion}(e)).

\paragraph{Dungeon illusion:} two gray rectangles are perceived as darker or lighter depending on the gray intensities of both the background and the grid, see \cite{bressan2001explaining}. The left rectangle is perceived as darker than the one on the right (Fig. \ref{fig:summary_brightness_illusion}(f)).

\paragraph{Grating induction:} the background grating (which can be tuned to different orientations) induces the appearance of a counter-phase grating in the homogeneous gray horizontal bar \cite{mccourt1982spatial} (Fig. \ref{fig:summary_brightness_illusion}(g)).
%A sinewave luminance grating induces the appearance of a counterphase sinusoidal grating in the homogeneous central gray region. The background grating can be tuned to different orientations.

\paragraph{Hong-Shevell illusion:} the mid-gray half-ring on the left appears darker than the one on the right, while both are identical \cite{hong2004brightness} (Fig. \ref{fig:summary_brightness_illusion}(h)).
%semi-circles of the same gray intensity are immersed in an alternating black and white semi-circles pattern on the background, see \cite{hong2004brightness}. As a result, the gray semi-circle lying between those black appears brighter than the other. See Fig. \ref{fig:summary_brightness_illusion}, third on the bottom from left to right.

\paragraph{Luminance illusion:} four identical dots over a background where intensity increases from left to right, and the dots on the left are perceived being lighter than the ones on the right \cite{kitaoka} (Fig. \ref{fig:summary_brightness_illusion}(i)).

\begin{figure}[hbtp]
\centering			
\begin{subfigure}[t]{0.19\linewidth}
    \centering\includegraphics[width=\linewidth]{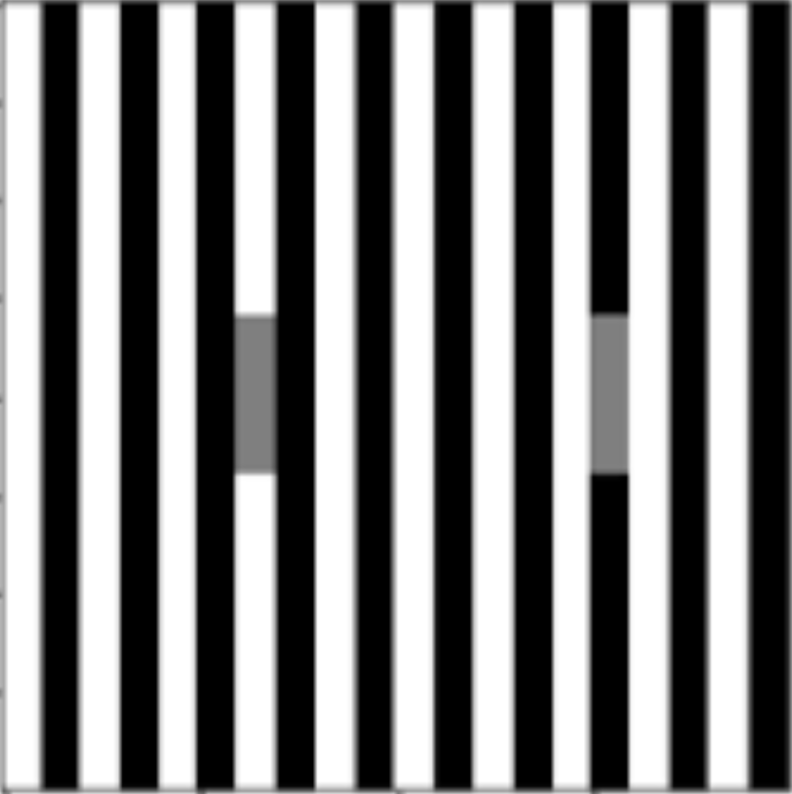}
    \caption{White}
  \end{subfigure}
  \begin{subfigure}[t]{.19\linewidth}
    \centering\includegraphics[width=\linewidth]{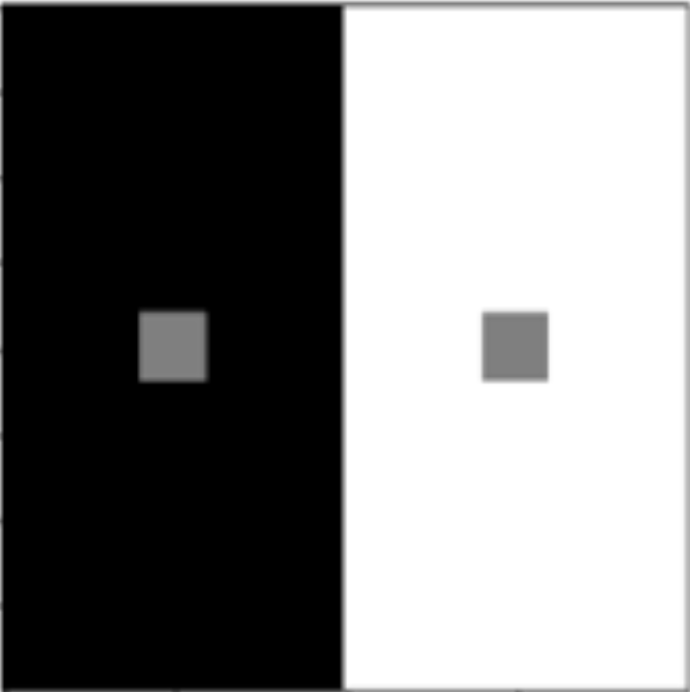}
    \caption{Brightness contrast}
  \end{subfigure}
  \begin{subfigure}[t]{.19\linewidth}
    \centering\includegraphics[width=\linewidth]{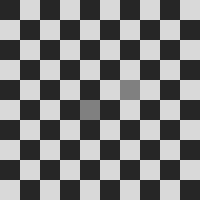}
    \caption{Checkerboard}
  \end{subfigure}
    \begin{subfigure}[t]{.19\linewidth}
    \centering\includegraphics[width=\linewidth]{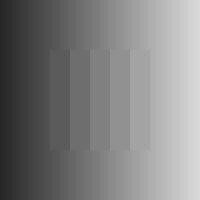}
    \caption{Chevreul}
  \end{subfigure}
    \begin{subfigure}[t]{.19\linewidth}
    \centering\includegraphics[width=\linewidth]{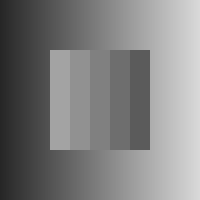}
    \caption{Chevreul cancellation}
  \end{subfigure}
  \begin{subfigure}[t]{0.19\linewidth}
    \centering\includegraphics[width=\linewidth]{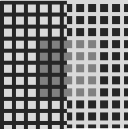}
    \caption{Dungeon}
  \end{subfigure}
  \begin{subfigure}[t]{.19\linewidth}
    \centering\includegraphics[width=\linewidth]{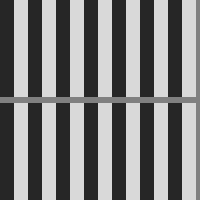}
    \caption{Grating induction}
  \end{subfigure}
  \begin{subfigure}[t]{.19\linewidth}
    \centering\includegraphics[width=\linewidth]{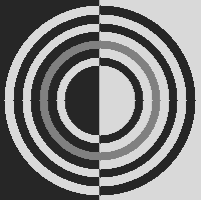}
    \caption{Hong-Shevell}
  \end{subfigure}
    \begin{subfigure}[t]{.19\linewidth}
    \centering\includegraphics[width=\linewidth]{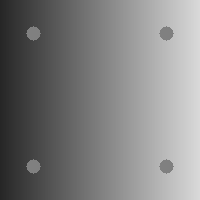}
    \caption{Luminance}
  \end{subfigure}
\caption{From left to right, top to bottom: White's illusion, Brightness contrast, the Checkerboard illusion, the Chevreul illusion, Chevreul cancellation, the Dungeon illusion, the Grating induction, the Hong-Shevell illusion and the Luminance illusion.}
\label{fig:summary_brightness_illusion}
\end{figure}

\subsubsection{Orientation-dependent illusions}
We also consider
orientation-dependent illusions,
where the perceptual phenomenon (e.g. in terms of brightness or contrast) is affected by the orientation of the image elements.

\paragraph{Poggendorff illusion.}
The Poggendorff illusion, {presented in the modified version considered in this work} in Fig. \ref{fig:summary_orientation_illusion}(a), is a very well known geometrical optical illusion in which the presence of a central surface induces a misalignment of the background lines. This illusion depends both on the orientation of the background lines and the width of the central surface \cite{Weintraub1971}, as the more the angle is close to $\pi/2$ the less is the bias, but in this example the perceived bias is also dependent on the brightness contrast between central surface and background lines.
\paragraph{Tilt illusion.}
The Tilt illusion is a phenomenon where the perceived orientation of a test line or grating is altered by the presence of surrounding lines or a grating with a different orientation. In our case we consider the effect that the orientation of a surround grating pattern 
has
on the perceived contrast of a grating pattern in the center: the inner circles in Figs. \ref{fig:summary_orientation_illusion}(b) 
and 
\ref{fig:summary_orientation_illusion}(c)
are identical but the latter is perceived as having more contrast than the former.
%Whereas its classical version is devoid of contrast differences between center and surroundings, in the two versions presented in Fig. \ref{fig:summary_orientation_illusion} (center, right), see  \cite{Self2014}, a difference of brightness within concentric circles is considered in order to test the bias dependency on both features. Furthermore, the main idea is to observe an enhanced bias when the brightness difference is followed by a change within the gratings orientation of the concentric circles.
\begin{figure}[hbtp]
\centering

\begin{subfigure}{0.25\textwidth}
\centering
    \includegraphics[height=3.5cm]{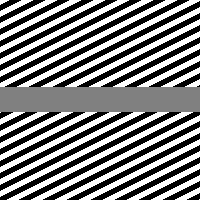}
    \caption{Poggendorff illusion.}
     \label{fig:poggendorff}
 \end{subfigure}
\begin{subfigure}{0.25\textwidth}
\centering
\includegraphics[height=3.5cm]{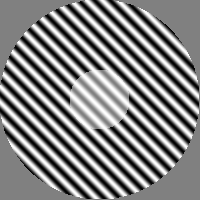}
\caption{Tilt illusion, same $\theta$.}
\label{fig:tilt_illusion-same}
\end{subfigure}
\begin{subfigure}{0.25\textwidth}
\centering
\includegraphics[height=3.5cm]{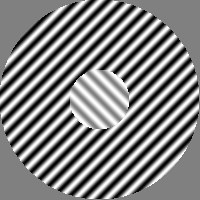}
\caption{Tilt illusion, different $\theta$.}
\label{fig:tilt_illusion-different}
\end{subfigure}
\caption{From left to right: a modified version of the Poggendorff illusion based on Grating Induction, a modified Tilt illusion with concentric circles having the same orientation and a modified Tilt illusion with concentric circles having different orientations.}
\label{fig:summary_orientation_illusion}
\end{figure}
%\begin{figure}[t]
%\centering
%\begin{subfigure}{0.48\textwidth}
%\centering
%Poggendorff illusion: the presence of a central surface induces a misalignment of the background lines, where the relative orientation of the background lines is $\theta=\pi/3$.
   % \includegraphics[height=4cm]{images/White_illusion.png}
    %\caption{White illusion: two grey rectangles having the same colour and opacity appear as having different grey-level intensities due to the underlying background. The left box seems darker than the right one.}
    % \label{fig:white_illusion}
    %Grating Induction with relative background orientation tuned to $\theta=\pi/2$. A sinewave luminance grating induces the appearance of a counterphase sinusoidal grating in a homogeneous central region. The background grating can be tuned to different orientations.
% \end{subfigure}\begin{subfigure}{0.48\textwidth}
%\centering
   % \includegraphics[height=4cm]{images/Brightness_contrast.png}
    %\caption{Brightness contrast: the different backgrounds induce a change in brightness perception of the grey central boxes, although they have the same colour and opacity.}
%     \label{fig:brightness_contrast}
% \end{subfigure}

%%%%%%%%%%%%%%%%%%%%%%%%%%%%%%%%%%%%%%%
%%%%%%%%%%%%%%%%%%%%%%%%%%%%%%%%%%%%%%%
%%%%%%%%%%%%%%%%%%%%%%%%%%%%%%%%%%%%%%%
%%%%%%%%%%%%%%%%%%%%%%%%%%%%%%%%%%%%%%%
%%%%%%%%%%%%%%%%%%%%%%%%%%%%%%%%%%%%%%%

\subsection{Wilson-Cowan-type models for contrast perception}

In this section we introduce four different evolution equations derived from the Wilson-Cowan formulation, that will be studied in this paper.
We recall that, denoting by $a(x,t)$ the state of a population of neurons with spatial coordinates $x\in\mathbb R^2$ at time $t>0$, the Wilson-Cowan equations proposed in \cite{WilsonCowan1973,Wilson1973bis} can be written\footnote{In \cite{WilsonCowan1973} the sigmoid function is applied outside of the integral term and not only on the activity $a(y,t)$ as in \eqref{eq:WC}. This corresponds to an ``activity-based'' model of neuron activation, while \eqref{eq:WC} corresponds to a ``voltage-based'' one. See \cite{Faugeras2009}, where the two models are shown to be equivalent.} as
\begin{equation}\label{eq:WC}
    \frac{\partial}{\partial t} a(x,t) = -\beta a(x,t) + \nu \int_{\mathbb R^2} \omega(x \| y) \sigma(a(y,t))\,dy + h(x),
\end{equation}
where $\beta>0$ and $\nu\in\R$ are fixed parameters, $\sigma:\R\to \R$ is a non-linear sigmoid saturation function, the kernel $\omega(x\|y)$ models interactions at two different spatial locations $x$ and $y$ (we will assume that the integral of $\omega$ is normalised to $1$) and $h$ is the input signal.

\subsubsection{Wilson-Cowan equations do not fulfill any variational principle}

Over the last thirty years, the use of variational methods in imaging has become increasingly popular as a regularisation strategy for solving general ill-posed imaging problems in the form
\begin{equation}  \label{eq:inv}
\text{find }u \quad\text{s.t.}\quad f=\mathcal{T}(u).
\end{equation}
Here, $f$ represents a given degraded image and $\mathcal{T}$ a (possibly non-linear) operator describing the degradation (e.g. noise, blur, under-sampling, etc.)

Due to the lack of fundamental properties such as existence, uniqueness and stability of the solution of the problem \eqref{eq:inv}, the idea of regularisation consists of incorporating \emph{a priori} information on the desired image $u_\star$ and on its closeness to the data $f$ by means of suitable variational terms.
This gives rise, in particular, to variational methods where one looks for an approximation $u_\star$ of the real solution $u$ by solving
\begin{equation}  \label{eq:prob_var}
  u_\star = \arg\min \mathcal E(u),
\end{equation}
where $\mathcal E$ is the energy functional combining regularisation and data fit, depending also on the given image $f$. 
A popular way to solve the variational problem
consists in finding $u_\star$ as the steady-state solution of the evolution equation given by the gradient descent of the energy functional
\begin{equation}\label{eq:grad_desc}
  \frac{\partial}{\partial t} u = - \nabla\mathcal E(u), \qquad u|_{t=0} = f,
\end{equation}
under appropriate conditions on the boundary of the image domain. 
 
In the context of vision science, evolution equations have been originally used as a tool to describe the physical transmission, diffusion and interaction phenomena of stimuli in the visual cortex \cite{Beurle1956,WilsonCowan1973,Wilson1973bis}.
Variational methods are the main tool of ecological approaches, that pose the efficient coding problem \cite{Olshausen2000} as an optimisation problem to be solved with evolution equations that minimise an energy functional \cite{Atick1992} involving natural image statistics and biological constraints. The resulting solution is optimal because it has minimal redundancy.

%Similarly, variational methods have been studied by the vision community to understand visual adaptation phenomena and to describe \emph{efficient neural coding}, see \cite{Webster2015,Olshausen2000}.\textcolor{red}{These references do indeed deal with efficient coding, but I don't recall they use variational methods for this: could you please double-check? (Marcelo).} 

However, we must remark that, while considering the gradient descent of an energy functional gives always an evolution equation, the reverse is not true: not every evolution equation is minimising an energy functional.
In fact, this is the case for the Wilson-Cowan equations, which do not fulfil any variational principle, as we prove in Appendix~\ref{a:non-var}.
As a consequence, they are sub-optimal in reducing the redundancy.

We remark that it is possible to define an energy that decreases along trajectories of \eqref{eq:WC}, as done in \cite{French2004}. This  ensures in particular that even though the evolution is not variational, its steady states (i.e.,  solutions of \eqref{eq:WC} that are constant in time) can indeed be obtained as critical points of this energy.

\subsubsection{A modification of the Wilson-Cowan equations complying with efficient representation}  \label{sec:efficientWC}
 
Remarkably, the efficient representation principle has correctly predicted a number of neural processing aspects and phenomena like
the photoreceptor response performing histogram equalisation, 
the dominant features of the receptive fields of retinal ganglion cells (lateral inhibition, the switch from bandpass to lowpass filtering when the illumination decreases, and, remarkably, colour opponency, with photoreceptor signals being highly correlated but color opponent signals having quite low correlation),
%{\color{red}, with $L$, $M$ and $S$ signals \textcolor{magenta}{ADD DEFINITION} being highly correlated but $L+M$, $L-M$ and $S-(L+M)$ having quite low correlation}), 
or the receptive fields of cortical cells having a Gabor function form \cite{Atick1992,Daugman1985a,Olshausen2000}. 
Efficient representation is the only framework able to predict the functional properties of neurons from a simple principle, and given how simple the assumptions are it's really surprising that this approach works so well \cite{Meister1999}.

In \cite{BertalmioJPP2009} it is shown how a slight modification of the Wilson-Cowan formulation leads to a variational model, as we now present. Assuming that the activity signal $a$ is in the range $[0,1]$, we can re-write equation \eqref{eq:WC} in terms of a sigmoid $\hat\sigma$ shifted by $\frac{1}{2}$ (which we take as the average signal value) and inverted in sign, thus getting: 
\begin{equation}\label{eq:WC2}
    \frac{\partial}{\partial t} a(x,t) = -\beta a(x,t) - \nu \int_{\mathbb R^2} \omega(x \| y)  \hat\sigma\left(a(y,t)-\frac{1}{2}\right)\,dy + h(x).
\end{equation}
Note that this is just a re-writing of equation \eqref{eq:WC}, so it is still not associated to any variational method. However, if we now assume $\hat\sigma$ to be odd and replace the $\frac{1}{2}$ term by $a(x,t)$, we obtain
\begin{equation}\label{eq:LHE}
    \frac{\partial}{\partial t} a(x,t) = -\beta a(x,t) + \nu \int_{\mathbb R^2} \omega(x \| y)  \hat\sigma(a(x,t) - a(y,t))\,dy + h(x),
\end{equation}
and this equation is now a gradient descent equation, as it does fulfil a variational principle.

Furthermore, under the proper choice of parameters $\beta, \nu$ and input signal $h$, this evolution equation performs local histogram equalisation (LHE) \cite{Bertalmio2007}. This is key for our purposes, since, as 
Atick points out \cite{Atick1992}, one of the main types of redundancy or inefficiency in an information system like the visual system happens when
some neural response levels are used more frequently than others, and for this type of redundancy the optimal code is the one that performs histogram equalisation.

It is therefore expected that the modification of the Wilson-Cowan equations in \eqref{eq:LHE}, which better complies with the efficient representation principle, should be more effective in reducing redundancy than the original Wilson-Cowan model of equation \eqref{eq:WC}.

\subsubsection{Accounting for orientation}

Models \eqref{eq:WC} and \eqref{eq:LHE}
ignore orientation and as such they are not well-suited to explain a number of visual phenomena. 
For this reason, following \cite{BCFFP_SSVM}, we extend them to a third dimension, representing local image orientation, as follows.
We let $La: Q \times [0,\pi) \to \R$ be the cortical activation in V1 associated with the signal $a$, so that $La(x,\theta)$ encodes the response of the neuron with spatial preference $x$ and orientation preference $\theta$ to $a$. Mathematically, such activation is obtained via a suitable convolution with the receptive profiles of V1 neurons, as explained in Appendix~\ref{a:cortical}, see also \cite{Duits2010,Petitot,Prandi2017,Citti2006, sarti2015constitution}. Then, denoting  by $A(x,\theta,t)$ the cortical response at time $t$ for any $t>0$, the natural extension of equations  \eqref{eq:WC} and \eqref{eq:LHE} to the orientation dependent case is given by the two models:
\begin{gather}\label{eq:WC3D}
\begin{split}
  \frac{\partial}{\partial t} A(x,\theta,t)
  = -\beta A(x,\theta,t)+ \nu\int_0^\pi\int_{Q} \omega(x,\theta\|y,\phi)\sigma\Big(A(y,\phi,t)\Big)\,dy\, d\phi 
  + Lh(x,\theta), 
  \end{split}\\
\label{eq:LHE3D}
  \begin{split}
  \frac{\partial}{\partial t} A(x,\theta,t)= -\beta A(x,\theta,t)
  +\nu\int_0^\pi\int_{Q} \omega(x,\theta\|y,\phi) \hat\sigma\big(A(x,\theta,t)-&A(y,\phi,t) \big)\,dy\,d\phi+Lh(x,\theta), 
  \end{split}
\end{gather}
where $Lh(x,\theta)$ denotes the cortical activation in V1 corresponding to the visual input $h$ at spatial location $x$ and orientation preference $\theta$. 
We remark that these models describe the dynamic behaviour of activations in the 3D space of positions and orientation. As explained in Appendix~\ref{a:cortical}, once a stationary solution is found, the two-dimensional perceived image can be found by simply applying the formula 
\begin{equation}\label{eq:proj}
  a(x) = \frac1\pi\int_0^\pi A(x,\theta)\,d\theta.
\end{equation}

\subsubsection{Models under consideration}

We summarise here the four models we are going to test in the following sections. 
The orientation-independent WC and LHE models are:
\begin{align}
\frac{\partial}{\partial t} a(x,t)
& = -(1+\lambda)a(x,t)+ \frac{1}{2M}\int_{Q} \omega(x,y)\sigma\left(a(y,t)\right)\,dy + \lambda f_0(x)+\mu(x) \label{eq:WC2Dtag} \tag{WC-2D}\\ 
\frac{\partial}{\partial t} a(x,t)
& = -(1+\lambda)a(x,t) +\frac{1}{2M}\int_{Q} \omega(x,y)\hat\sigma\left(a(x,t)-a(y,t)\right)\,dy + \lambda f_0(x)+\mu(x) \label{eq:LHE2Dtag} \tag{LHE-2D},    
\end{align}
which relate to \eqref{eq:WC} and \eqref{eq:LHE} by simply choosing parameters as $\beta = 1+\lambda$ and  $\nu=1/2M$ where $M>0$ is a normalisation constant, and input signal $h(x)=\lambda f_0(x)+ \mu(x)$, where $\lambda>0$, $f_0(x)$ is the local intensity at $x\in Q$ of given image $f_0$ and $\mu(x)$ denotes a local average  of the initial stimulus $f_0$ around $x$ (a choice motivated by the averaging behaviour of cells in the magnocellular pathway \cite{Bertalmio2019VisualModels} and already considered in similar models e.g.  \cite{Bertalmio2007,BertalmioFrontiers2014}).

The orientation-dependent WC and LHE models can be similarly written as:
\begin{align}
\frac{\partial}{\partial t} A(x,\theta,t)
  = & -(1+\lambda)A(x,\theta,t)  + \frac{1}{2M}\int_0^\pi\int_{Q} \omega(x,\theta||y,\phi)\sigma\left(A(y,\phi,t)\right)\,dy\,d\phi \notag \\  & + \lambda Lf_0(x,\theta)+L\mu(x,\theta), \label{eq:WC3Dtag} \tag{WC-3D}\\ 
\frac{\partial}{\partial t} A(x,\theta,t)
 = & -(1+\lambda)A(x,\theta,t)+\frac{1}{2M}\int_0^\pi\int_{Q} \omega(x,\theta||y,\phi)\hat\sigma\left(A(x,\theta,t)-A(y,\phi,t)\right)\,dy\,d\phi  \notag \\
 & + \lambda Lf_0(x,\theta)+L\mu(x,\theta), \label{eq:LHE3Dtag} \tag{LHE-3D} 
\end{align}
which can analogously be related to \eqref{eq:WC3D} and \eqref{eq:LHE3D} by choosing the very same parameters as above and by now taking as cortical activation in V1 corresponding to $h$ the quantity $Lh(x,\theta)=\lambda Lf_0(x,\theta)+L\mu(x,\theta)$.

\subsubsection{Numerical implementation} 
All four relevant equations \eqref{eq:WC2Dtag}, \eqref{eq:LHE2Dtag}, \eqref{eq:WC3Dtag}, and \eqref{eq:LHE3Dtag} are numerically implemented via a forward Euler time-discretisation, as presented in \cite{Bertalmio2007}. For a given image $a$, the cortical activation $La$ is recovered via standard wavelet transform methods, as presented in \cite{BCFFP_SSVM} (see also \cite{Duits2010}). The codes, written in Julia \cite{bezanson2017julia}, are available at the following link: http://www.github.com/dprn/WCvsLHE.

All the considered images are of size $200\times 200$ pixels, and take values in the interval $[.15,.85]$ in order to avoid out-of-range issues.  We always consider $K = 30$ discretised orientations, as done in \cite{Boscain2018} for instance. As presented in Appendix~\ref{a:cortical}, the receptive profiles associated to the discretised orientations selected are obtained via cake wavelets \cite{Bekkers2014}, for which the frequency band \texttt{bw} is set to \texttt{bw}$ = 5$. 
The interaction kernel is taken to be a $2$D or $3$D Gaussian with standard deviation $\bm{\sigma}_\omega$, the local mean average $\mu$ is obtained via Gaussian filtering with standard deviation $\bm{\sigma}_\mu$. 
In our experiments we used the following two piece-wise linear functions as sigmoids:
\begin{equation}  \label{def:sigmoid}
  \hat\sigma(\rho) := \min\{1,\max\{\alpha\rho, -1\}\}, \qquad \quad \sigma(\rho) :=-\hat\sigma\left(x-\frac{1}{2}\right),
\end{equation}
with $\alpha = 5$, see Figure~\ref{fig:sigmoids}. Note that $\hat\sigma$, which will be used for LHE models, is odd and centered in zero while $\sigma$, which will be used for WC models, is shifted in $1/2$ and shows a reversed behaviour. This in fact corresponds to a change of sign in the integral terms of LHE models w.r.t.\ the WC ones, as discussed in Section \ref{sec:efficientWC}.

%As previously mentioned, depending on the sign of $\nu$, model \eqref{eq:WC} is able to describe both excitatory ($\nu >0$) and inhibitory local interactions ($\nu<0$), see, e.g. \cite[Section~3]{Bressloff2002}. Due to the oddness of $\sigma$, this latter case can be equivalently expressed by keeping $\nu>0$ and replacing $\sigma$ with its ``mirrored'' version $\hat\sigma(x) = \sigma(-x)$. (See Figure~\ref{fig:sigmoids}.)

\begin{figure}[htp]   
    \centering
    \begin{subfigure}{0.45\textwidth}
    \centering
    \includegraphics[width=\textwidth]{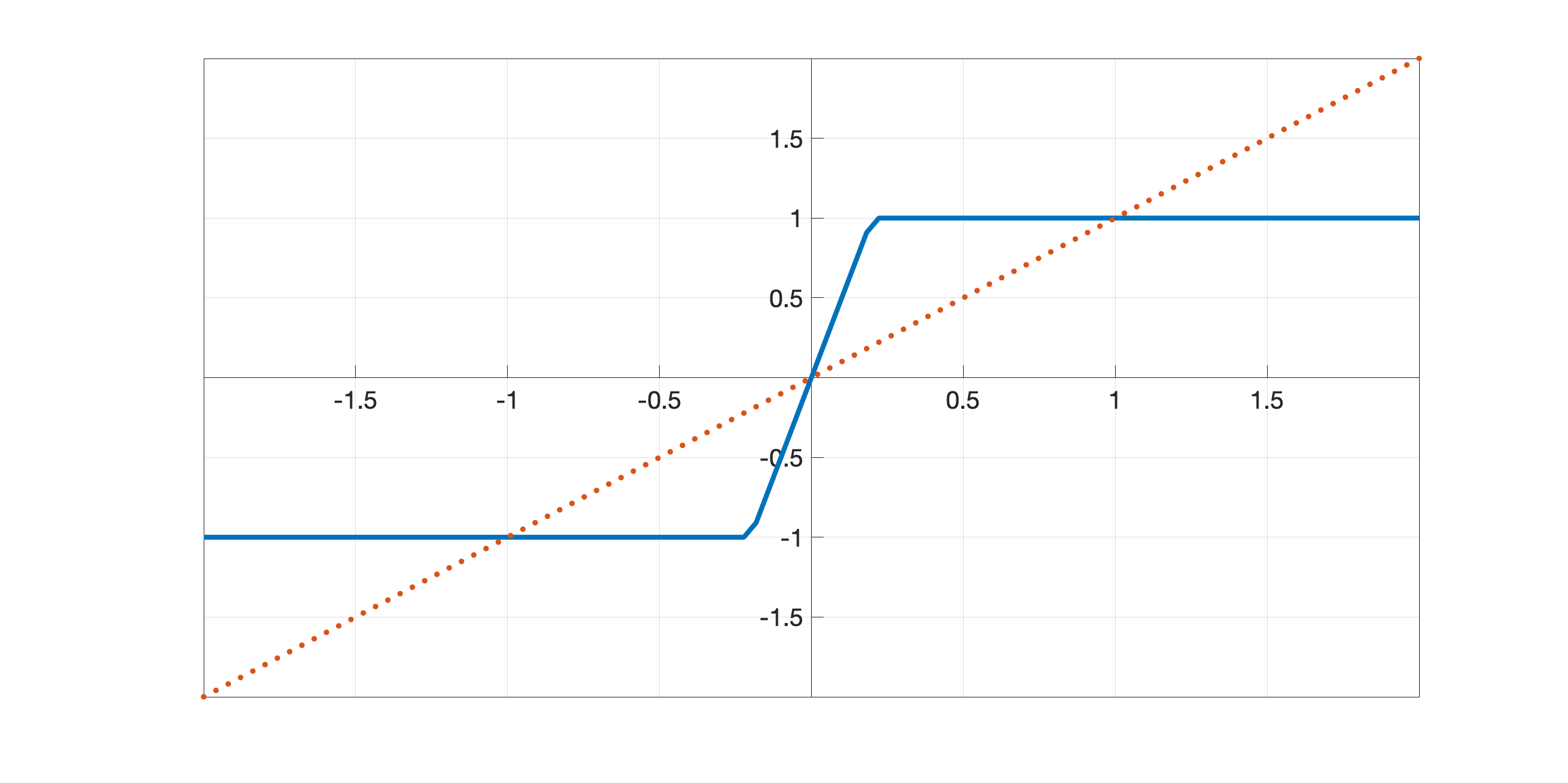}
    \caption{$\hat\sigma$ and the line $y=x$}
    \end{subfigure}
    \begin{subfigure}{0.45\textwidth}
    \centering
    \includegraphics[width=\textwidth]{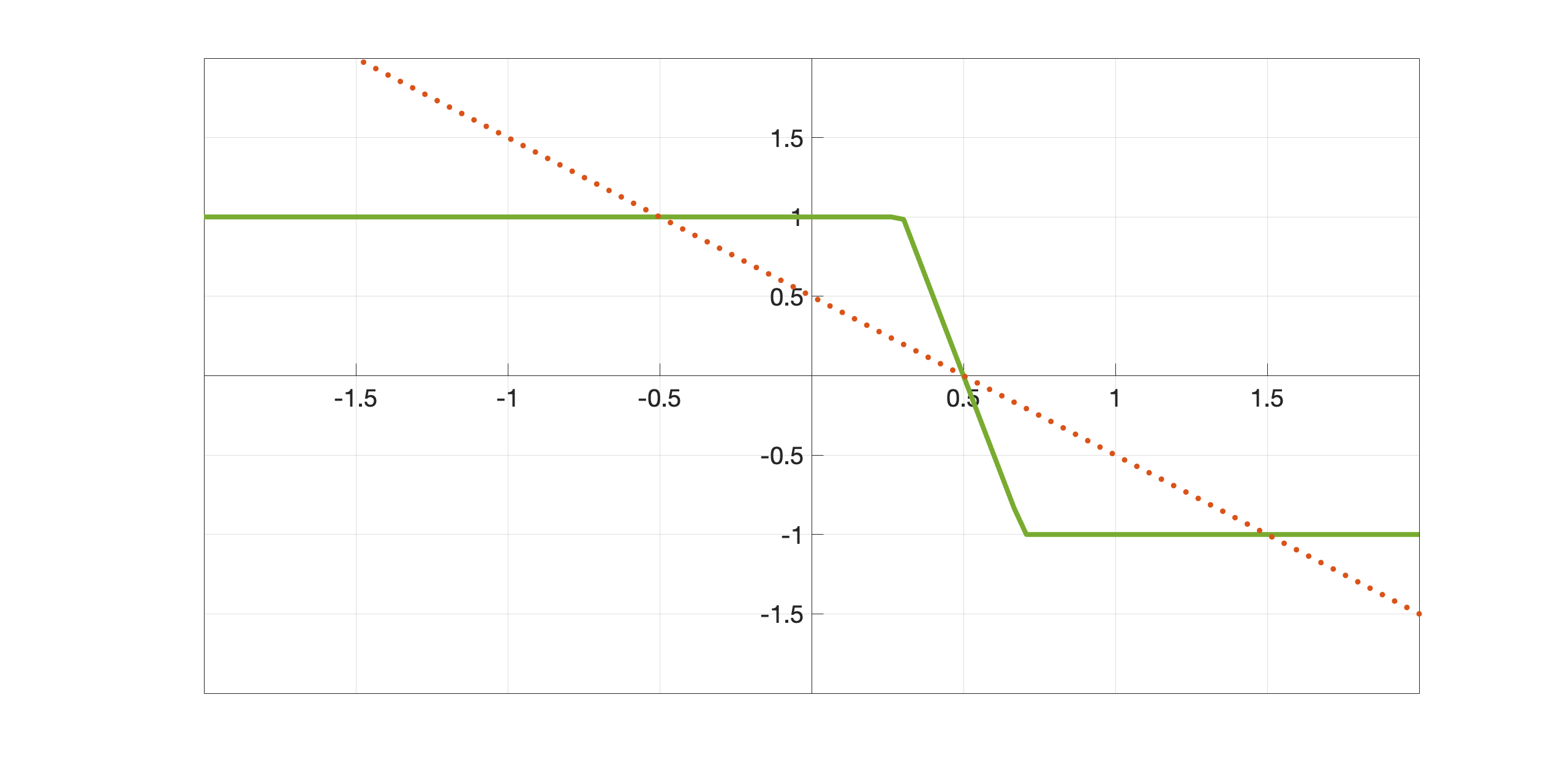}
    \caption{$\sigma$ and the line $y=-x+\frac{1}{2}$}
    \end{subfigure}
    \caption{Sigmoid functions in the form \eqref{def:sigmoid}, with $\alpha=5$, as considered in our experiments. }
      \label{fig:sigmoids}
\end{figure}

Finally, the evolution stops when the $L^2$ relative distance between two successive iterations is smaller than a tolerance $\tau = 10^{-2}$, which identifies convergence of the iterates to a stationary state.

\section{Results} \label{sec:results}

In this section, we present the results obtained by applying the four models described above to the visual illusions described in Section \ref{sec:visual_ill}. Our objective is to understand the capability of these models to \emph{replicate} the visual illusions under consideration. That is, we are interested in whether the output produced by the models qualitatively agrees with the human perception of the phenomena.
We stress that our study is purely qualitative as it has to be intended as a proof of concept showing how Wilson-Cowan-type dynamics can be effectively used to replicate the perceptual effects due to the observation of visual illusions. We do not address here the match with empirical data since those depend on several experimental conditions for which a correspondence with the model parameters is not clear. A dedicated study on experiments motivated by psychophysics, addressing the validation of our models and, possibly, allowing for the creation of ground-truth references for a quantitative assessment is left for future research.

Due to the lack of a universal metric adapted to the task of assessing the replication of visual illusions, we will evaluate replication or lack thereof by presenting relevant line profiles, i.e., plots of brightness levels along a single row (line), of images produced by the four models in consideration (a common tool used by several brightness/lightness/color models before ~\cite{Blakeslee2016,Otazu2008}).
These lines are chosen as to cross a section of the image called \emph{target}: A gray region in the image (or set of regions in the case of the Chevreul illusion), 
%whose shape may be a rectangle, a square or a circle 
where the brightness illusion appears. 
%A quantitative analysis of the results is extremely challenging, due to the lack of a universal notion of contrast used in this setting for validate the computational tests.

In all the results shown in this section, the original visual stimulus profile is represented as a blue dashed line. The line profiles of the output models are represented as solid red \eqref{eq:LHE2Dtag}, green \eqref{eq:WC2Dtag}, magenta \eqref{eq:LHE3Dtag}, and cyan \eqref{eq:WC3Dtag} lines.

%Profiles are extracted in the rows that cross a section of the image called \emph{target}. A target is a grey region in the image (or set of regions in the case of Chevreul illusion), whose shape may be a rectangle, a square or a circle where the brightness illusion appears as a result of perception bias.

%For each illusion and each model, we found the parameters (or the set of parameters) by which the replication of the illusion is possible. Table \ref{t:param} specifies some specific instances of these parameters used for each model to compute the output images presented in the section.

The parameters appearing in the models have been chosen independently for each illusion and each model, in order to obtain the best possible replication of the visual illusion.  Here, by best-replication we mean that the extracted line-profiles correctly mimic the perceptual outcome from a qualitative point of view.The chosen parameters are presented in Table~\ref{t:param}.\newpage

\newcolumntype{g}{>{\columncolor{Gray}}c}
\newcolumntype{a}{>{\columncolor{Gray}}l}
\begin{table}[ht]
\centering
\begin{tabular}{a||g|g|g|g||g|g|g|g||g|g|g|g||g|g|g|g}
\hline
\rowcolor{white}
& \multicolumn{4}{c||}{\textbf{WC-2D}} & \multicolumn{4}{c||}{\textbf{LHE-2D}} & \multicolumn{4}{c||}{\textbf{WC-3D}} & \multicolumn{4}{c}{\textbf{LHE-3D}}\\
\hline
\rowcolor{white}
\textbf{Illusion} &
$\bm{\sigma}_\mu$ & $\bm{\sigma}_\omega$ & $\lambda$ & $M$ & 
$\bm{\sigma}_\mu$ & $\bm{\sigma}_\omega$ & $\lambda$ & $M$ & 
$\bm{\sigma}_\mu$ & $\bm{\sigma}_\omega$ & $\lambda$ & $M$ & 
$\bm{\sigma}_\mu$ & $\bm{\sigma}_\omega$ & $\lambda$ & $M$ \\
\hline
\rowcolor{white}
White
 & 10 & 20 & .7 & 1.4
 & 10 & 50 & .7 & 1
 & 20 & 30 & .7 & 1.4
 & 2 & 50 & .7 & 1
\\
\hline
Brightness
& 2 & 10  & .7 & 1.4 
& 2 & 10  & .7 & 1 
& 2 & 10  & .7 & 1.4 
& 2 & 10  & .7 & 1 
\\
\hline
\rowcolor{white}
Checkerboard
 & 10  & 70 & .7 & 1.4
 & 10  & 70 & .7 & 1
  & 10  & 70 & .7 & 1.4
   & 10  & 70 & .7 & 1
   \\
\hline
Chevreul
 & 2  & 5 & .7 & 1
 & 2 & 10 & .7 & 1
 & 2 & 40 & .5 & 1
 & 5 & 7 & .7 & 1
\\
\hline
\rowcolor{white}
Chevreul canc.
 & 2 & 2 & .9 & 1
 & 5 & 3 & .9 & 1
 & 2 & 20 & .5 & 1.4
 & 5 & 3 & .9 & 1
\\
\hline
Dungeon
 & 6  & 10 & .7 & 1.4
 & 5 & 40 & .7 & 1
 & 2 & 50 & .7 & 1.4
 & 5 & 50 & .7 & 1
\\
\hline
\rowcolor{white}
Gratings
 & 2 & 6 & .7 & 1
 & 2 & 6 & .7 & 1
 & 2 & 6 & .7 & 1
 & 2 & 6 & .7 & 1
\\
\hline
Hong-Shevell
 & 5 & 20 & .7 & 1
 & 5 & .5 & .7 & 1
 & 10 & 30 & .7 & 1
 & 10 & 30 & .7 & 1
\\
\hline
\rowcolor{white}
Luminance
 &  2 & 6 & .7 & 1
 &  2 & 6 & .7 & 1
 &  2 & 6 & .7 & 1
 &  2 & 6 & .7 & 1
\\
\hline
Poggendorff
 & \xmark & \xmark & \xmark &\xmark
 & \xmark & \xmark & \xmark &\xmark
 & \xmark & \xmark & \xmark &\xmark
  & 3 & 10  & .5 & 1
\\
\hline
\rowcolor{white}
Tilt
 & \xmark & \xmark & \xmark &\xmark
 & \xmark & \xmark & \xmark &\xmark
 & \xmark & \xmark & \xmark &\xmark
  & 15 & 20 & .7 & 1
\\
\hline
\end{tabular}
\caption{Parameters used in the tests. %Here we let $\nu = 1/(2M)$. 
%\red{VALE: in the code that we have right now, the sigmoid for the WC equations is $\sigma/2$. The value $M$ in the table hence has to be doubled to be consistent with the presentation of the present paper (but these are the values to be inserted for the moment in the code).}
}
\label{t:param}
\end{table}

%The response of each model will be measure in the target (except in Chevreul's effects). A target is a gray region which can be a rectangle, square or a circle where the illusion (change of brightness) appear.

%Additionally the following term will we 

%\paragraph{Replication.} We say that a model replicates a VI if the difference between the image levels of the targets at the output produced by the model agrees qualitatively with human perception.

%\paragraph{Target.} In all the selected visual illusions (except Chevreul's effects) there is a gray region, which can be a rectangle, square or a circle. In this section we will refer to the region as target since inside this region the brightness change (illusion) will appear.

%In order to measure the successfully reproduction of perception we will take the profiles of the target region. A profile is a plot of a single row of an image. In all the profiles presented in this section, the original visual stimulus profile is presented as a blue dash line. The row used for the profile is stated in the result of each illusion.

%\paragraph{Profile.} A profile is a plot of a single row of an image. In all the profiles presented in this section the original visual stimulus profiles is presented as a blue dash line. The row used in each profile is stated in the result of each illusion.

%\paragraph{Assimilation and contrast.} Two basic opposite brightness effects. In assimilation the target intensity goes towards the value in the surrounding regions. Contrast for the contrary, makes the target goes further from the neighbourhood regions.

\subsection{Orientation-independent brightness illusions}

{%\color{blue} 
Table \ref{t:param} summarises the replication results obtained for the illusions described in Section \ref{sec:visual_ill}: if the model replicates the illusion we indicate in the table the used parameters, otherwise a cross (\xmark) denotes no replication, i.e. the failure of the model to reproduce computational results corresponding to the visual perception of the considered illusion.}

\paragraph{White's illusion.} 
The chosen line profile for the plots in Fig.~\ref{fig:resultWhite} corresponds to the central horizontal line of the image, which crosses both gray patches.
%The plots in Fig.~\ref{fig:resultWhite} show the profiles of the middle row in the image. The blue dashed  line show the intensity profile of the visual stimulus evaluated in correspondence of the middle row of the image, thus containing both grey patches. 
%The solid lines represents the brightness intensity predicted by  each model (left: 2D models, right: 3D models). 
As both plots show, all four models correctly predict the left target to be darker than the right one.

\begin{figure}[h!]
\centering			
    \centering\includegraphics[width=0.81\linewidth]{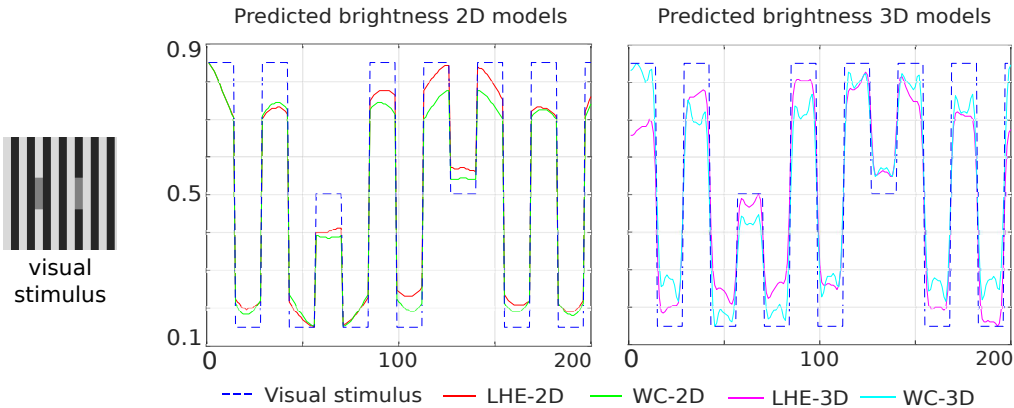}
\caption{Predicted brightness in White illusion}
\label{fig:resultWhite}
\end{figure}

\paragraph{Simultaneous brightness contrast.}  %A gray rectangle is perceived darker when it is in front of a bright background, and brighter when it is in front of a dark background  (see visual stimulus in Fig.~\ref{fig:resultBright}).

The plots in Fig.~\ref{fig:resultBright} show the line profiles of the central horizontal line of the image, which crosses the two gray squares. We see that our four models replicate this illusion (left square lighter than the right square). In both the 2D and the 3D case, we observe that LHE methods result in an enhanced contrast effect w.r.t. WC methods.

\begin{figure}[h!]
\centering			
    \centering\includegraphics[width=0.81\linewidth]{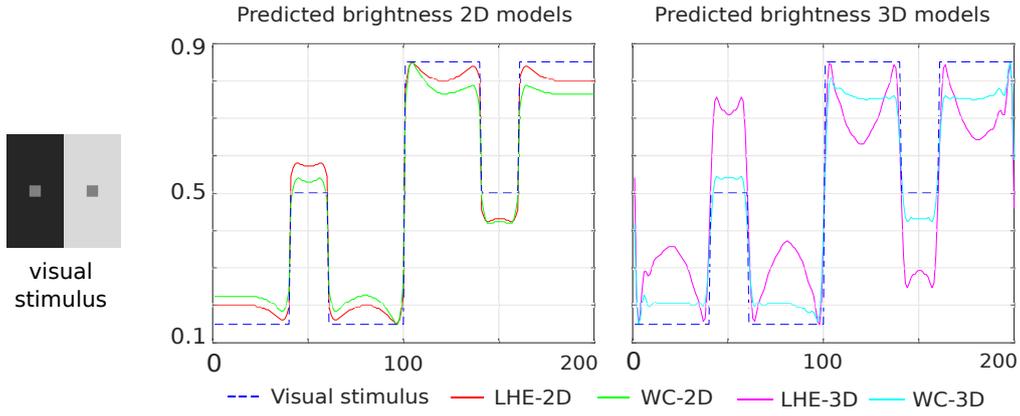}
\caption{Predicted brightness in simultaneous brightness contrast}
\label{fig:resultBright}
\end{figure}

\paragraph{Checkerboard illusion.} %Two gray rectangles are embedded in a chessboard pattern. The rectangle with black neighbours (vertical and horizontal) looks darker than the rectangle with white neighbours (see visual stimulus in Fig.~\ref{fig:resultCheck}).

The chosen line profiles for this illusion are the two horizontal lines crossing, respectively, the left gray target and the right one. In Fig.~\ref{fig:resultCheck}, we chose to plot the first half of the line profile corresponding to the left target and the second half of the one corresponding to the right target.
%
%The plots in Fig.~\ref{fig:resultCheck} show two different rows of the image outputs. The first one shows an image row intersecting the right grey target, while the second one corresponds to a row intersecting the left grey target. 
The profiles of all the four models show replication of this illusion, by which the left target is perceived darker than the right one.

\begin{figure}[hbtp]
\centering			
    \centering\includegraphics[width=0.81\linewidth]{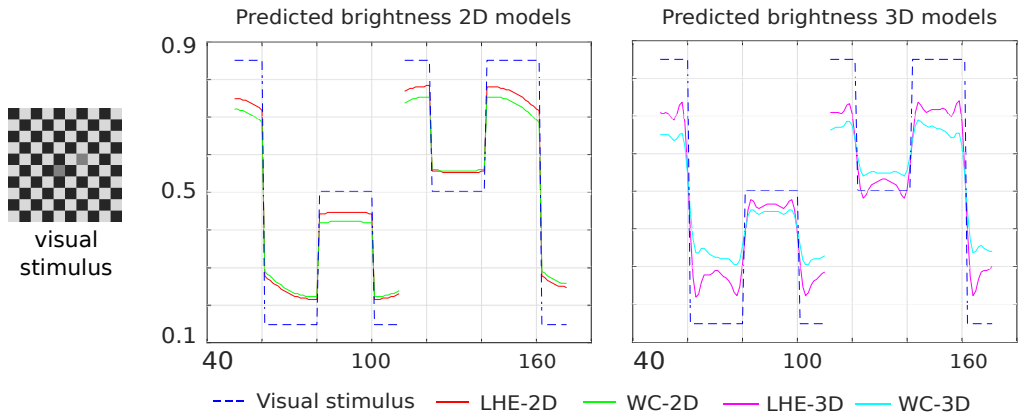}
\caption{Predicted brightness in Checkerboard illusion}
\label{fig:resultCheck}
\end{figure}
\newpage
\paragraph{Chevreul illusion.} Fig.~\ref{fig:resultChe} presents the line profiles for the central horizontal line.All four models correctly replicate the perceived changes within each band.

\begin{figure}[hbtp]
\centering			
    \centering\includegraphics[width=0.76\linewidth]{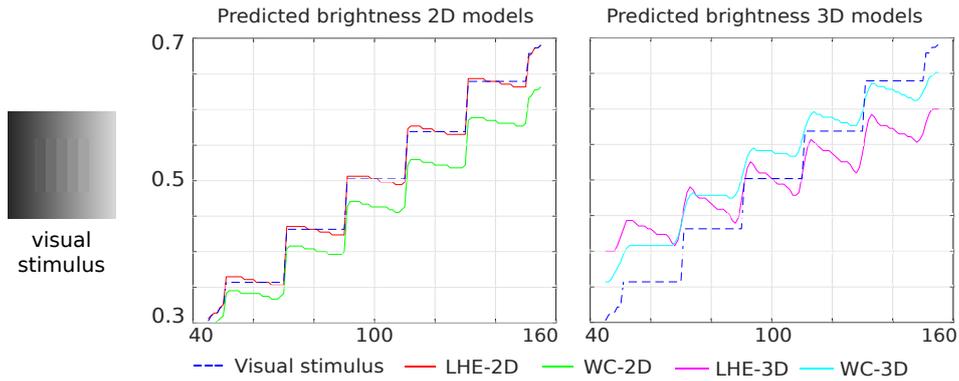}
\caption{Predicted brightness in Chevreul illusion}. 
\label{fig:resultChe}
\end{figure}

%A pattern of homogeneous bands of increasing intensity from left to right is presented (see visual stimulus in Fig.~\ref{fig:resultChe}). However, the bands  in the image are perceived to be in-homogeneous, with darker and brighter lines at the borders between adjacent bands.

%Profiles of middle rows are presented in Fig.~\ref{fig:resultChe}.  

%and the `crimping' effect close to the edges. 
%Both in the 2D and in the 3D case, the WC method has ae with respect to visual stimulus level.

\paragraph{Chevreul cancellation.} %This stimulus is similar to Chevreul illusion but with and inverted pattern of increasing bands (now from right to left). Here, there is not illusion or in-homogeneous effect as in Chevreul illusion (see visual stimulus in Fig.~\ref{fig:resultCheCan}).
%\red{Change comment}

The line profiles for the central horizontal line are presented in  Fig.~\ref{fig:resultCheCan}. In this case all models are able to correctly replicate the effect, although in the case of  \eqref{eq:WC2Dtag} and \eqref{eq:LHE3Dtag} the perceptual response is not perfect, due to the presence of some oscillations. We also remark that the correct replication of this illusion is extremely sensitive to the chosen parameters.

%{\color{red}
%In this case, WC-2D fail to replicate the envenom since there is an under and over-shoot at each intensity change (see green line in Fig.~\ref{fig:resultCheCan}). LHE-2D is capable of replicating this illusion as the flat regions in the red line shown. In the case of 3D models both are capable of replicate the effect, however, in the case of LHE-3D the perceptual response is not perfect (since some oscillation exists in the flat regions).
%}

%In this case three of our proposed models fail to replicate human perception. The response to this visual stimulus should be an homogeneous staircase, similar to the input profile. Nevertheless, The predicted response of WC-3D shows  homogeneous bands (see cyan line in right plot), hence it is successfully replicating the illusion. To the best of our knowledge this is the first brightness perception model capable of replicate this effect. % not sure if we should put this claim

\begin{figure}[hbtp]
\centering			
    \centering\includegraphics[width=0.81\linewidth]{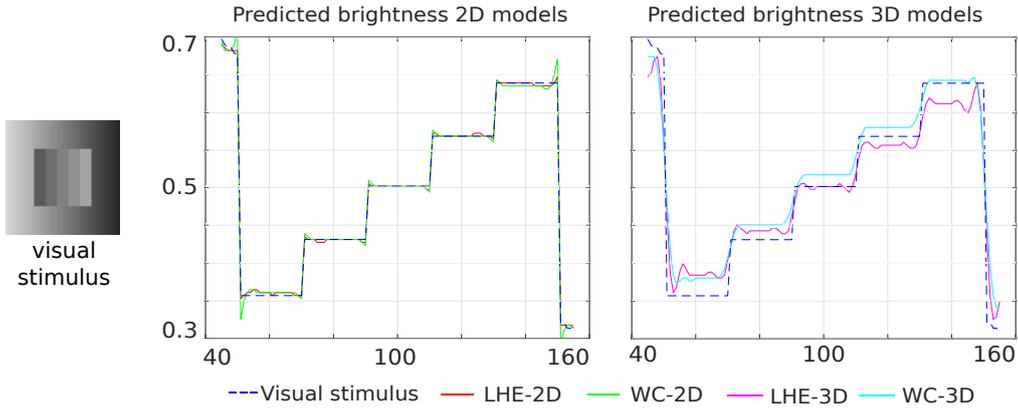}
\caption{Predicted brightness in Chevreul cancellation}
\label{fig:resultCheCan}
\end{figure}

\paragraph{Dungeon illusion.} %As visual stimulus in Fig.~\ref{fig:resultDung} shows, two gray rectangles are perceived as darker of lighter depending in the background and a grid imposed over them. The rectangle in the white background and black grid is perceived darker than the rectangle in dark background and white grid.

Profiles of the central section (3 middle squares) of each target are shown in Fig.~\ref{fig:resultDung}. The first part of the plot (left to right) represents the profile of the rectangle on black background. The second plot shows the target on white background.
As these profiles show, our four proposed models replicate human perception (first target is predicted as darker than the second). Nevertheless, the assimilation effect (target intensity goes towards surrounding) is stronger in the 3D models.

\begin{figure}[h!]
\centering			
    \centering\includegraphics[width=0.79\linewidth]{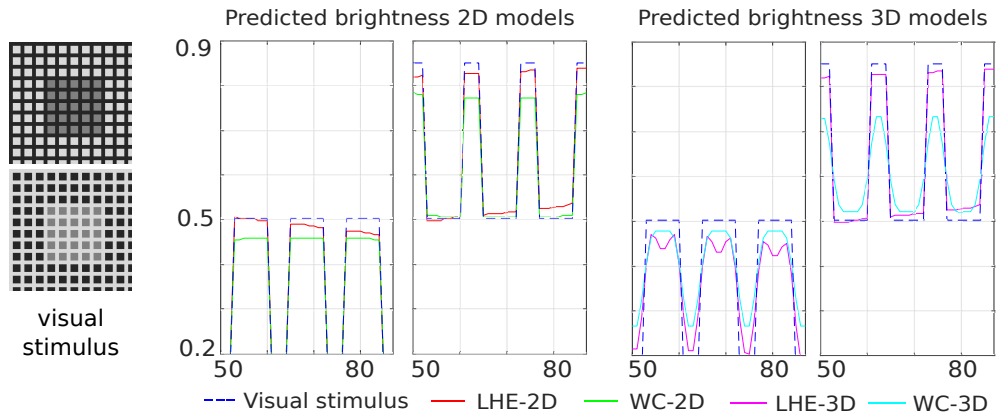}
\caption{Predicted brightness in Dungeon illusion}
\label{fig:resultDung}
\end{figure}

\paragraph{Grating induction.} %This illusion is composed of a grating with a gray horizontal stripe in the middle (see visual stimulus in Fig.~\ref{fig:resultGrat}). Observers report to see a grating in counter-phase to the original one in the gray middle line. 

In Fig.~\ref{fig:resultGrat} the continuous and dashed blue lines respectively show the profile of the grating and of the central horizontal line row of the visual stimulus. Then, the line profiles of the central horizontal line of the outputs have been plotted.
%The continuous lines (red, green, magenta, and cyan) represent the middle row profile of the respose to the proposed models. 
%
We observe that for both 2D and 3D models a counter-phase grating appears in the middle row, which successfully coincides with human perception. Notice that LHE methods have a higher amplitude in both cases.

\begin{figure}[h!]
\centering			
    \centering\includegraphics[width=0.79\linewidth]{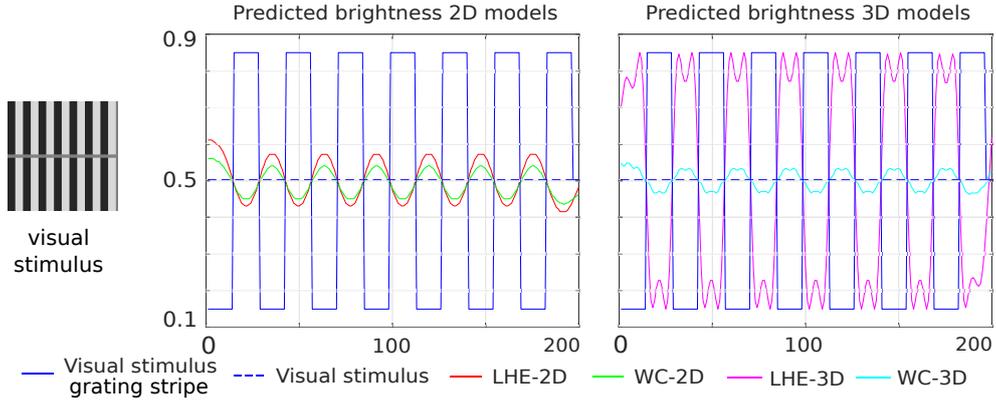}
\caption{Predicted brightness in grating induction}
\label{fig:resultGrat}
\end{figure}

\paragraph{Hong-Shevell illusion.}  Fig.~\ref{fig:resultHS} shows the line profiles of the central horizontal line around the target (gray ring) neighbourhood rings in the first half of the image. 
As in the case of the Dungeon illusion, we present in the first half of the plot (left to right) the output of the first stimulus (light background) and in the second half the output of the second (dark background).
%Hence, first plot (left to right) represent the first stimulus, where a gray ring surrounded by white rings is shown.
%
%In the four plots in Fig.~\ref{resultHS} 
We see how our four proposed models replicate the assimilation effect. Hence, the gray ring in the first image is predicted as brighter than the gray ring in the second visual stimulus. 

\begin{figure}[hbtp]
\centering			
    \centering\includegraphics[width=0.79\linewidth]{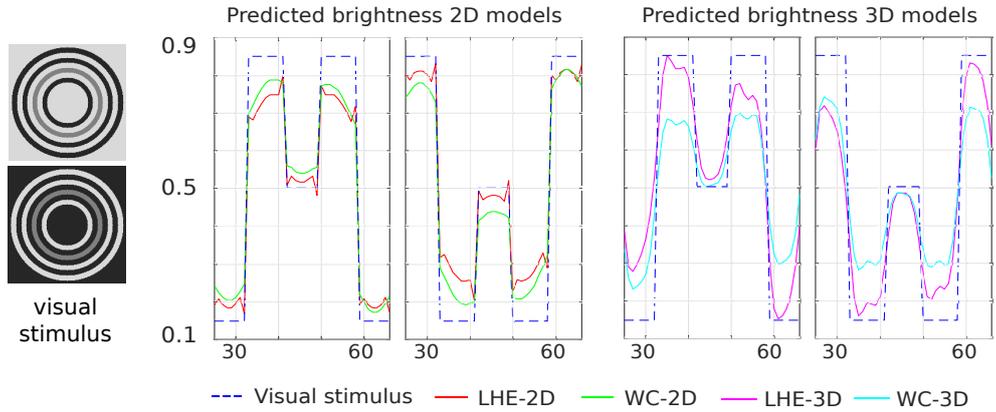}
\caption{Predicted brightness in Hong-Shevell illusion}
\label{fig:resultHS}
\end{figure}

\newpage

\paragraph{Luminance illusion.} Horizontal profiles crossing top left and right targets (gray circles) are depicted in Fig.~\ref{fig:resultLum}. For each target our four models reconstruct the left target as brighter than the right one. Hence,  all models correctly predict this contrast effect. In this case, LHE presents a higher contrast response in both responses (2D and 3D).

\begin{figure}[hbtp]
\centering			
    \centering\includegraphics[width=0.78\linewidth]{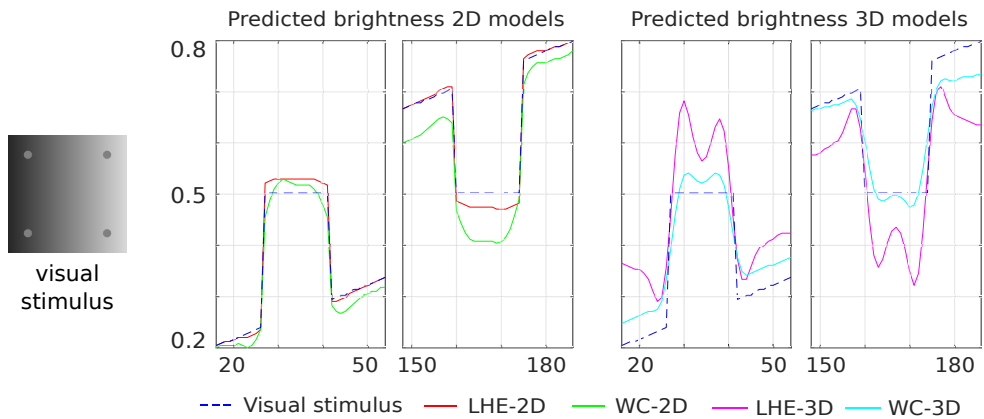}
\caption{Predicted brightness in luminance gradient illusion}
\label{fig:resultLum}
\end{figure}

%We observe that in all the considered brightness-dependent illusions the amplitude of the effect is always higher for the LHE methods (2D and 3D). Additionally, the LHE-3D method presents neighbourhood-dependent oscillations which can or not happen in the WC-3D method.

We observe that in all the considered brightness illusions both the 3D methods present neighbourhood-dependent oscillations. 

\subsection{Orientation-dependent illusions}

%Table \ref{table:illusion_orientation} summarise the replication of the selected {\color{blue} orientation and brightness} illusions. %A checkmark indicates replication and a $x$-mark no coincidence with human perception. When the parameters of the model must fulfil a certain condition for replication, it is indicated in parenthesis.

%\newcolumntype{g}{>{\columncolor{Gray}}c}
%\begin{table}[ht]
%\centering
%\begin{tabular}{g|g|g|g|g}
%\hline
%\rowcolor{white}
%\textbf{Illusion} & \textbf{WC-2D} & %\textbf{LHE-2D} & \textbf{WC-3D} & %\textbf{LHE-3D}\\
%\hline
%\rowcolor{white}
%Poggendorff
% & \xmark
% & \xmark 
% & \xmark 
% & \cmark %($\bm{\sigma}_\mu=2,3,5,10,20, \bm{\sigma}_\omega=10,20$) 
% \\
%Tilt
% & \xmark
% & \xmark 
% & \xmark 
% & \cmark %($\bm{\sigma}_\mu=15, %\bm{\sigma}_\omega=20$)
% \\
%\hline
%\end{tabular}
%\caption{Replication of orientation-dependent visual illusions via the proposed models.}% Possible parameter restrictions are indicated in each cell.}
%\label{table:illusion_orientation}
%\end{table}

%\red{Results + line profiles}
\paragraph{Poggendorff illusion.} 
%\textcolor{red}{ Fig.~\ref{fig:resultPog} shows a profile of the middle row in the visual stimulus, while the} 
The output images and a zoom of the target gray middle area are presented in Fig.~\ref{fig:resultPogIn}. In this case \eqref{eq:WC2Dtag}, \eqref{eq:WC3Dtag}, and \eqref{eq:LHE2Dtag} are not able to completely replicate the illusion, since induced white lines on the gray area are not connected.
On the other hand, \eqref{eq:LHE3Dtag} successfully replicates the perceptual completion over the gray middle stripe. %(see Fig.~\ref{fig:resultPog} right plot). Fig.~\ref{fig:resultPogIn} shows a zoom in the output of the proposed methods,  inpainted black lines over the gray regions can be observed in the LHE-3D output.  

%1) WC2D: the phenomenon is weekly replicated for all combinations of parameters, up to sigmaMu < sigmaW (see wc2d_pogg_grat_fig1 - line profiles). For sigmaMu > sigmaW doesn’t replicate.

%2) LHE2D: the phenomenon is not replicated, as the algorithm produces impainting of the gratings bar (lhe2d_sM_2_sW10).  Parameters SigmaMu >> SigmaW  (e.g. 10 vs 2) produces instability of the lines profiles. For SigmaMu == SigmaW (5,5), line profile is flat (not impainting nor completing). For SigmaMu and sigmaW both big and SigmaMu<SigmaW lines profile are flat (not impainting nor completing). [lhe2d_lines_profiles_SigmaW20]

%3) WC3D:  SigmaW (small = 2, 5), it doesn’t replicate (not completing). For SigmaMu = 3, it replicates when SigmaMu < SigmaW, but weakly. Weak completion also for SigmaM << sigmaW (2,20).

%4) LHE3D: SigmaW (small = 2), the phenomenon is not replicated visually. SigmaMu (=2,3,5,10,20) and SigmaW = 5, it performs impainting [lhe3d_sM_3_sW_5]. SigmaWu = 10 and SigmaW = 20 it strongly replicated the illusion (completion).

%\begin{figure}[hbtp]
%\centering			
%    \centering\includegraphics[%width=0.81\linewidth]{images/pog_fig.eps}
%\caption{Predicted brightness in Poggendorff illusion}
%\label{fig:resultPog}
%\end{figure}

\begin{figure}[hbtp]
\centering			
    \centering\includegraphics[width=0.9\linewidth]{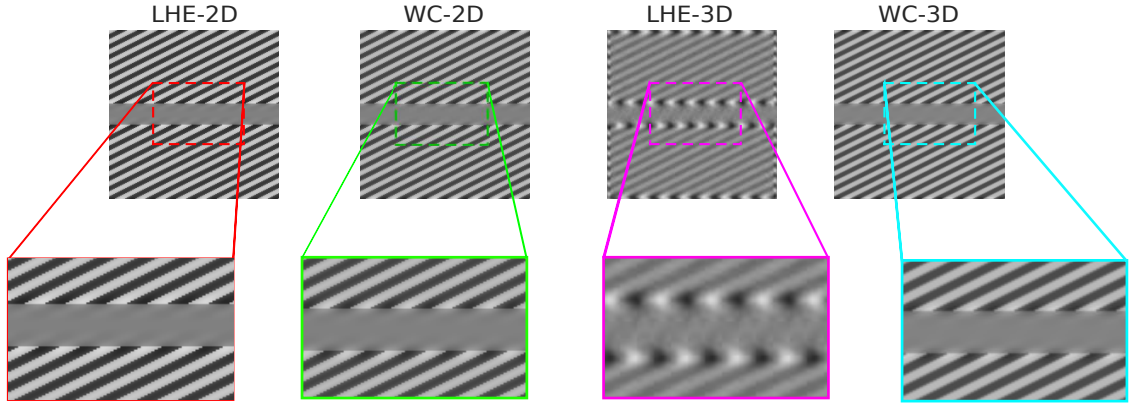}
\caption{Zoom of the predicted completion for Poggendorff illusion}
\label{fig:resultPogIn}
\end{figure}

\paragraph{Tilt illusion.} 
In Fig.~\ref{fig:resultTilt} we present line profiles, for both visual stimuli, for a diagonal line starting at the bottom left corner of the image and ending at the top right one. In order to be able to correctly compare the two images, the line profile of the second image (from top to bottom) has been extracted after flipping the outer circle along the vertical axis, so that the responses to both stimulus have the same background. %Then, the profile was obtained from a line starting in the bottom left and ending in the top right.  
%
%The first plot in Fig.~\ref{fig:resultTilt} (left to right) shows the plot of the first stimulus (top to bottom) and the second plot shows the profile of the second visual stimulus. 
Although there is a noticeable effect, such as a reduction in contrast for the \eqref{eq:WC2Dtag}, the difference between the responses to the two stimuli is very mild for all models with the exception of \eqref{eq:LHE3Dtag}. % evidence some visibility effect.

The fact that indeed this model is replicating the effect can be better appreciated looking at Fig.~\ref{fig:TiltShow}, which shows a composite of the inner circle for the responses to the two visual stimuli of the two orientation-dependent models. It is then evident that the \eqref{eq:LHE3Dtag} model yields a stronger result than the \eqref{eq:WC3Dtag} one. In fact, the former shows increased visibility (measured here as the contrast) for the half of the circle corresponding to the second stimulus  than the one corresponding to the first stimulus. On the other hand, in the case of the \eqref{eq:WC3Dtag} model (or of 2D models, not depicted here), the circle shows no difference among its two halves. This justifies our claim that the \eqref{eq:LHE3Dtag} model can increase the visibility of the inner circle (replicate the illusion) based on the orientation of the outer circle.

\begin{figure}[hbtp]
\centering			
    \centering\includegraphics[width=0.81\linewidth]{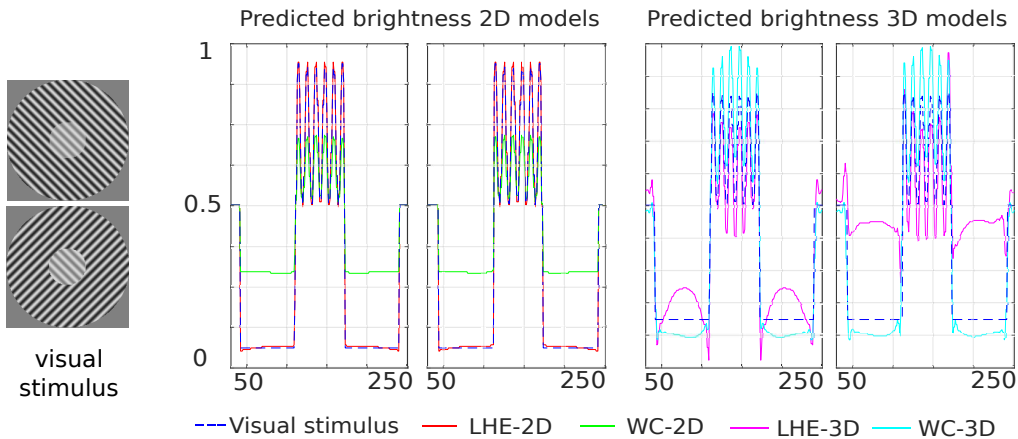}
\caption{Predicted brightness in Tilt illusion}
\label{fig:resultTilt}
\end{figure}

\begin{figure}[hbtp]
\centering			
    \centering\includegraphics[width=0.85\linewidth]{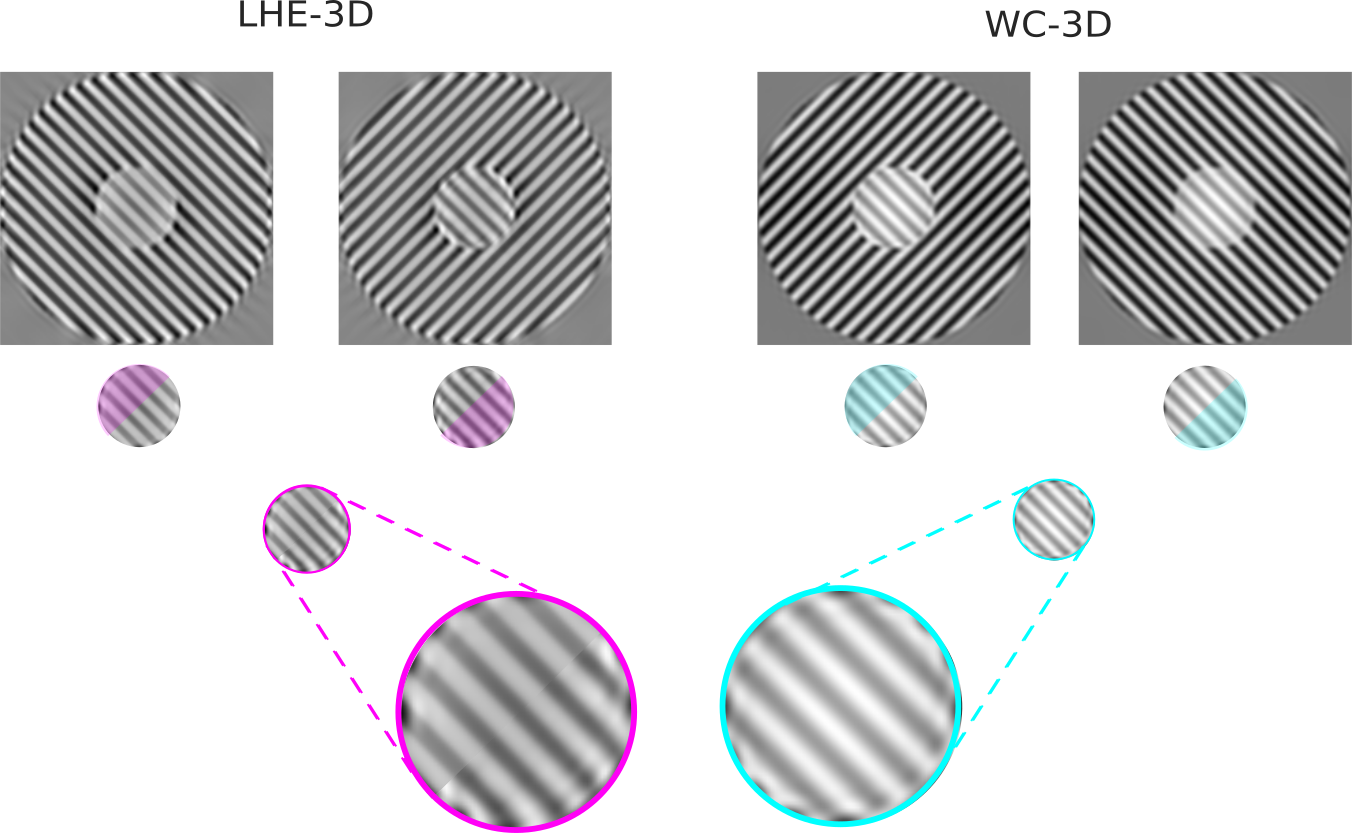}
\caption{Detail in predicted brightness in Tilt illusion}
\label{fig:TiltShow}
\end{figure}

\section{Discussion}

The results presented in the previous section show that the four models are able to reproduce several brightness illusions. Concerning orientation-dependent illusions we observe that, as expected, 2D models cannot reproduce them, while the only 3D model that correctly reproduces the perceptual outcome is the \eqref{eq:LHE3Dtag}. However we stress that determining replication or lack thereof in the Tilt illusion is subtle, as the observed effects are very mild.

As already mentioned, the parameters of the presented results are chosen independently from one illusion to the other in order to qualitatively optimise the perceptual replication in terms of suitable line profiles. Empirical observations show that the value of the model parameters involved are indeed related with the size of the target and the spatial frequency of the background. Nevertheless, if one settles for milder replications, it would be possible to choose more uniform parameters. For instance, this happens for the \eqref{eq:WC3Dtag} model in the Chevreul and Chevreul cancellation illusions, which can be reproduced simultaneously with parameters $\bm{\sigma}_\mu = 3$ and $\bm{\sigma}_\omega = 30$, although with less striking results.

Regarding the 3D models, we want to point out that we have chosen to use $K=30$ orientations whereas this number commonly takes values in the 12-18 range in the literature (e.g. \cite{Scholl2013,Chariker2016,Pattadkal2018}). 
Our selection of 30 orientations is motivated by some preliminary tests (which we are not presenting here) showing that a coarser orientation discretisation seems not to be sufficient to reproduce most of  the orientation-dependent illusions.  As future research we will test whether or not a different selection of parameters allows to reproduce those illusions with less orientations, but we should also mention that some works in the literature actually use a high number of orientations in cortical models (e.g. 64 orientations in \cite{Teich2010}).  
%Regarding the 3D models, as far as the number of orientations is considered, we notice that the choice $K=30$ is higher than the ones reported in  \cite{}, where only 12-15 orientations are usually considered. Our choice is motivated by some preliminary tests (which we are not presenting here) showing that a coarser orientation discretisation seem not to be sufficient to reproduce most of  the orientation-dependent illusions by means of the 3D algorithms described above.  A more accurate analysis of our algorithms run in these scenarios, possibly coupled with a careful parameter selection, should be addressed in future research.
%Alex commented  - Finally, we notice that the output of 3D models often shows oscillations. This is especially true for the \eqref{eq:LHE3Dtag} model. Such oscillations can be observed as dependent on the target surrounding.

Finally, we notice that the output of 3D models often shows oscillations. For some illusions (white and dungeon), the \eqref{eq:WC3Dtag} model produces more oscillatory solutions than \eqref{eq:LHE3Dtag}, and for others (Chevreul brightness, grating induction, and luminace gradient), the \eqref{eq:LHE3Dtag} have stronger oscillations than \eqref{eq:WC3Dtag}. The relation between the model parameters and possible dependence of the target surrounding is a matter of future research.

%This is especially true for the \eqref{eq:LHE3Dtag} model. Such oscillations can be observed as dependent on the target surrounding.

\section{Conclusions and future work}

We consider Wilson-Cowan-type models describing neuronal dynamics and apply them to the study of  replication of brightness visual illusions. 

We show that Wilson-Cowan equations are able to replicate a number of brightness illusions and that their variational modification, accounting for changes in the local contrast  and performing local histogram equalisation, outperforms them.
We consider also extensions of both models accounting for explicit local orientation dependence, in agreement with the architecture of V1. Although in the case of pure brightness illusions we found no real advantage in considering models taking into account orientations, these turned out to be necessary for the replication of  two exemplary orientation-dependent illusions, which only the 3D LHE variational model is able to reproduce.

In order to understand and fully exploit the potential of the orientation-dependent LHE model, further research should be done. In particular, a more accurate modelling reflecting the actual structure of V1 should be addressed. This concerns first the lift operation, where the cake wavelet should be replaced by the more physiologically plausible Gabor filters, as well as the interaction weight $\omega$ which could be taken to be the anisotropic heat kernel of \cite{Citti2006,sarti2015constitution,Duits2010}.
The design of appropriate psychophysics experiments testing the visual illusions considered in this work and their match with our models' outputs
is clearly a further important research direction, which would turn our qualitative study into a quantitative one.  The problem of matching computational models of perception with psychophysical data is in fact not trivial, but necessary to provide insights about how visual perception works and to identify the computational parameters able to reproduce the perceptual bias induced by these phenomena.

\section*{Acknowledgements and Grants}

M.~B. would like to thank the organizers of the conference to celebrate Jack Cowan's 50 years at the University of Chicago for their kind invitation to attend that meeting, which served as inspiration for this work, and also acknowledges the support of the European Union’s Horizon 2020 research and innovation programme under grant agreement number 761544 (project HDR4EU) and under grant agreement number 780470 (project SAUCE), and of the Spanish government and FEDER Fund, grant ref. PGC2018-099651-B-I00 (MCIU/AEI/FEDER, UE).
L.~C., V.~F.\ and D.~P.\ acknowledge the support of a public grant overseen by the French National Research Agency (ANR) as part of the \emph{Investissement d'avenir program}, through the iCODE project funded by the IDEX Paris-Saclay, ANR-11-IDEX-0003-02 and of the research project \emph{LiftME} funded by INS2I, CNRS. 
L.~C.\ and V.~F.\ acknowledge the support provided by the \emph{Fondation Math\'ematique Jacques Hadamard}. V.~F.\ acknowledges the support received from the European Union's Horizon 2020 research and innovation programme under the \emph{Marie Sk\l odowska-Curie grant No 794592}.
V.~F.\ and D.~P.\ also acknowledge the support of ANR-15-CE40-0018 project \textit{SRGI - Sub-Riemannian Geometry and Interactions}. B.~F.\ acknowledges the support of the Fondation Asile des Aveugles.

% \section*{Grants}

\section*{Disclosures}
All authors declare that there is no commercial relationship relevant to the subject matter of presentation.

\section*{Author contributions}
All authors equally contributed to this work.

\appendix
\section{Non-variational nature of Wilson-Cowan equation}\label{a:non-var}

In this section we show that, for non-trivial choices of weight and sigmoid functions, Wilson-Cowan equations do not admit a variational formulation.

For the sake of simplicity, we will consider only a finite dimensional variant of Wilson-Cowan equations, with constant input. Namely, for $a:\mathbb R\to \mathbb R^n$, we consider
\begin{equation}\label{eq:wc-a}
  \frac{d}{dt} a(t) = -\mu a(t) + W\sigma(a(t))+h.
\end{equation}
Here, $h\in\R^n$ is the input, $\mu>0$ is a parameter, $\sigma\in C^1(\R)$ is any function (we denoted $\sigma(v)=(\sigma(v_i))_i$ for $v\in\R^n$), and $W\in \R^{n\times n}$ is a symmetric interaction kernel. For a proof in the infinite-dimensional setting we refer to \cite{BCFFP_JMIV}

Equation \eqref{eq:wc-a} admits a variational formulation if it can be written as the steepest descent associated with a functional $J:\R^n\to \R$, i.e.,
\begin{equation}\label{eq:steepest}
  \frac{d}{dt} a(t)= -\nabla J(a(t)).
\end{equation}
We have the following.

\begin{theorem}
%   Assume that \eqref{eq:steepest} holds. 
  The Wilson-Cowan equation \eqref{eq:wc-a} admits the variational formulation \eqref{eq:steepest} only if either $W$ is a diagonal matrix, or $\sigma$ is an affine function, i.e., $\sigma(x)=\alpha x+\beta$ for some $\alpha,\beta\in\R$. 
\end{theorem}

\begin{proof}
Writing \eqref{eq:wc-a} and \eqref{eq:steepest} componentwise, we find the following relation for $J$:
\begin{equation*}\label{eq:der}
  \partial_i J(v) = \mu v_i -\sum_{k}W_{\ell,k}\sigma(v_\ell)-h_i, \qquad v=(v_1,\ldots,v_n)\in\R^n, \quad i=1,\ldots,n.
\end{equation*}
By differentiating again the above, and letting $\delta_{ij}$ denote the Kroenecker delta symbol, we have
\begin{equation}
  \partial_{ij} J(v) = \mu \delta_{ij} -\sum_{k}W_{\ell,k}\sigma'(v_\ell)\delta_{j\ell} = \mu \delta_{ij} -W_{ij}\sigma'(v_j),\qquad i,j=1,\ldots,n.
\end{equation}
Namely, $\operatorname{Hess}J(v)=(\mu\delta_{ij}-W_{ij}\sigma'(v_j))_{ij}$. 
Assume that $W$ is not a diagonal matrix. Then, since both the Hessian matrix and $W$ are symmetric, by choosing $i\neq j$ such that $W_{ij}\neq0$ we get
\begin{equation}
  \sigma'(v_i) = \sigma'(v_j) \qquad \forall v\in\R^n.
\end{equation}
This clearly implies that $\sigma'$ is constant, thus showing that $\sigma$ must be an affine function.
\end{proof}

  We observe that the above reasoning does not apply to the LHE algorithm. Indeed, the discrete form of the latter is
  \begin{equation}\tag{LHE}
    \frac{d}{dt} a(t) = -\mu a(t) + \sum_{\ell}W_{i\ell}\sigma\big(a_i(t)-a_\ell(t)\big)+h.
  \end{equation}
  Then, the corresponding variational equation (for $\mu=0$ and $h=0$) is
  \begin{equation}
    \partial_i J(v) = -\sum_{\ell\neq i} W_{i\ell}\sigma(v_i-v_\ell), \qquad v\in\R^n. 
  \end{equation}
  This yields
  \begin{equation}
    \partial_{ji}J(v)= W_{ij}\sigma'(v_i-v_j),\qquad \text{for }v\in\R^n,\quad i\neq j.
  \end{equation}
  This does not contradict the symmetry of the Hessian, as $\sigma$ was chosen to be odd an thus $\sigma'$ is even. Indeed, we know by \cite{BertalmioCowan2009} that we can let
  \begin{equation}
    J(v) := \sum_{k,\ell} W_{k\ell}\Sigma(v_k-v_\ell),
  \end{equation}
  where $\Sigma$ is such that $\Sigma'= \sigma$.

\section{Encoding orientation-dependence via cortical-inspired models}  \label{a:cortical}

Orientation dependence of the visual stimulus is encoded via cortical inspired techniques, following e.g., \cite{Citti2006,Duits2010,Petitot,Prandi2017,Bohi2017}. 
The main idea at the base of these works goes back to the 1959 paper \cite{HubelWiesel59} by Hubel and Wiesel (Nobel prize in 1981) who discovered the so-called \emph{hypercolumn functional architecture} of the visual cortex V1.
%, encoding how neurons in V1 are sensitive to position and direction. 
Following \cite{HubelWiesel59}, each neuron $\xi$ in V1 detects couples $(x,\theta)$ where $x\in\R^2$ is a retinal position and $\theta$ is a direction at $x$. Orientation preferences $\theta$ are then organised in  hypercolumns over the retinal position $x$, see \cite[Section~2]{Petitot}. 

%In the mathematical modelling, each orientation preference $\theta$ is
%identified with an element of the real projective line $\mathbb P^1$ (represented as $[0,\pi]/\sim$, with $0\sim \pi$), and each neuron $\xi$ in $V1$ is identified with  an element of $\mathcal M=\mathbb R^2\times \mathbb P^1$. The hypercolumn functional architecture of V1 is hence modelled by a \emph{functional lifting} that associates to each image $f(x)$ a  function over $\mathcal M$.
%We remark that the resulting procedure thus consists of a linear lifting combined with WC-type evolution, which is non-linear due to the presence of the sigmoid $\sigma_\alpha$, see \eqref{def:sigmoid}.

Let $Q\subset \R^2$ be the visual plane. To a visual stimulus $f:Q\to [0,1]$ is  associated a cortical activation $Lf:Q\times[0,\pi)\to \R$ such that $Lf(\xi)$ encodes the response of the neuron $\xi=(x,\theta)$. Letting $\psi_\xi\in L^2(\R^2)$ be the receptive profile (RP) of the neuron $\xi$, such response is assumed to be given by
\begin{equation}
  Lf(\xi) = \left\langle \psi_\xi, f \right\rangle_{L^2(\R^2)} = \int_{Q} \overline{\psi_\xi(x)}f(x)\,dx.
\end{equation}
Motivated by neuro-phyisiological evidence, we assume that RPs of different neurons are ``deducible'' one from the other via a linear transformation. 
As detailed in \cite{Duits2010,Prandi2017}, see also \cite[Section~3.1]{BCFFP_SSVM}, this amounts to the fact that the linear operator $L:L^2(Q)\to L^2(Q\times[0,\pi))$ is a continuous wavelet transform (also called \emph{invertible orientation score transform}). That is, there exists a \emph{mother wavelet} $\Psi\in L^2(\bR^2)$ such that $Lf(x,\theta) = \big[f * (\Psi^*\circ R_{-\theta})\big] (x)$. Here, $f*g$ denotes the standard convolution on $L^2(\bR^2)$ and $R_{-\theta}$ is the counter-clock-wise rotation of angle $\theta$. Notice that, although images are functions of $L^2(\bR^2)$ with values in $[0,1]$, it is in general not true that $Lf(x,\theta)\in [0,1]$.

Concerning the choice of the mother wavelet, we remark that neuro-physiological evidence suggests that a good fit for the RPs is given by Gabor filters, whose Fourier transform is the product of a Gaussian with an oriented plane wave \cite{Daugman1985a}.  However, these filters are quite challenging to invert, and are parametrised on a bigger space than $\mathcal M$, which takes into account also the frequency of the plane wave and not only its orientation. For this reason, in this work we instead considered \emph{cake wavelets}, introduced in \cite{duits2007image, Bekkers2014}. These are obtained via a mother wavelet $\Psi^{\text{cake}}$ whose support in the Fourier domain is concentrated on a fixed slice, depending on the number of orientations one aims to consider in the numerical implementation. For the sake of integrability, the Fourier transform of this mother wavelet is then smoothly cut off via a low-pass filtering, see \cite[Section ~2.3]{Bekkers2014} for details. Observe, however, that, since we are considering orientations on $[0,\pi)$ and not directions on $[0,2\pi)$, we choose a non-oriented version of the mother wavelet, given by $\tilde\psi^{cake}(\mathbf{\omega}) + \tilde\psi^{cake}(e^{i\pi}\mathbf{\omega})$, in the notations of \cite{Bekkers2014}.

An important feature of cake wavelets is that, in order to recover the original stimulus from its cortical activation, it suffices to simply ``project'' the cortical activations along hypercolumns. This yields
\begin{equation}
    f(x) := \frac1\pi\int_{0}^\pi Lf(x,\theta)\,d\theta.
\end{equation}
This justify the assumption, implicit in equation \eqref{eq:proj}, that the projection of a cortical activation $F$ (not necessarily given by a visual stimulus) to the visual plane is given by
\begin{equation}
  PF(x) = \frac1\pi \int_{0}^\pi F(x,\theta)\,d\theta.
\end{equation}

\bibliographystyle{apalike} 
\bibliography{biblio}

\begin{thebibliography}{}

\bibitem[Adini et~al., 1997]{Adini1997}
Adini, Y., Sagi, D., and Tsodyks, M. (1997).
\newblock Excitatory--inhibitory network in the visual cortex: Psychophysical
  evidence.
\newblock {\em Proceedings of the National Academy of Sciences},
  94(19):10426--10431.

\bibitem[Atick, 1992]{Atick1992}
Atick, J.~J. (1992).
\newblock Could information theory provide an ecological theory of sensory
  processing?
\newblock {\em Network: Computation in Neural Systems}, 3(2):213--251.

\bibitem[Attneave, 1954]{Attneave1954}
Attneave, F. (1954).
\newblock Some informational aspects of visual perception.
\newblock {\em Psychological review}, 61(3):183.

\bibitem[Barlow et~al., 1961]{Barlow1961}
Barlow, H.~B. et~al. (1961).
\newblock Possible principles underlying the transformation of sensory
  messages.
\newblock {\em Sensory communication}, 1:217--234.

\bibitem[Bekkers et~al., 2014]{Bekkers2014}
Bekkers, E., Duits, R., Berendschot, T., and ter Haar~Romeny, B. (2014).
\newblock A multi-orientation analysis approach to retinal vessel tracking.
\newblock {\em JMIV}, 49(3):583--610.

\bibitem[Benucci et~al., 2013]{Benucci2013}
Benucci, A., Saleem, A.~B., and Carandini, M. (2013).
\newblock Adaptation maintains population homeostasis in primary visual cortex.
\newblock {\em Nature neuroscience}, 16(6):724.

\bibitem[Bertalm\'io, 2014]{BertalmioFrontiers2014}
Bertalm\'io, M. (2014).
\newblock From image processing to computational neuroscience: a neural model
  based on histogram equalization.
\newblock {\em Front. Comput. Neurosc.}, 8:71.

\bibitem[Bertalm\'io et~al., 2019a]{BCFFP_SSVM}
Bertalm\'io, M., Calatroni, L., Franceschi, V., Franceschiello, B., and Prandi,
  D. (2019a).
\newblock A cortical-inspired model for orientation-dependent contrast
  perception: A link with {W}ilson-{C}owan equations.
\newblock In Lellmann, J., Burger, M., and Modersitzki, J., editors, {\em Scale
  Space and Variational Methods in Computer Vision}, pages 472--484, Cham.
  Springer International Publishing.

\bibitem[Bertalm\'io et~al., 2019b]{BCFFP_JMIV}
Bertalm\'io, M., Calatroni, L., Franceschi, V., Franceschiello, B., and Prandi,
  D. (2019b).
\newblock Cortical-inspired {W}ilson-{C}owan-type equations for
  orientation-dependent contrast perception modelling.
\newblock arXiv preprint: \url{https://arxiv.org/abs/1910.06808}.

\bibitem[Bertalm{\'{i}}o et~al., 2007]{Bertalmio2007}
Bertalm{\'{i}}o, M., Caselles, V., Provenzi, E., and Rizzi, A. (2007).
\newblock {Perceptual color correction through variational techniques}.
\newblock {\em IEEE T. Image Process.}, 16(4):1058--1072.

\bibitem[Bertalm{\'\i}o and Cowan, 2009]{BertalmioJPP2009}
Bertalm{\'\i}o, M. and Cowan, J.~D. (2009).
\newblock Implementing the retinex algorithm with wilson--cowan equations.
\newblock {\em Journal of Physiology-Paris}, 103(1-2):69--72.

\bibitem[Bertalm\'io and Cowan, 2009]{BertalmioCowan2009}
Bertalm\'io, M. and Cowan, J.~D. (2009).
\newblock Implementing the retinex algorithm with {W}ilson{-}{C}owan equations.
\newblock {\em J. Physiol. Paris}, 103(1):69 -- 72.

\bibitem[Bertalm{\'\i}o et~al., 2017]{BertalmioVSS2017}
Bertalm{\'\i}o, M., Cyriac, P., Batard, T., Martinez-Garcia, M., and Malo, J.
  (2017).
\newblock The wilson-cowan model describes contrast response and subjective
  distortion.
\newblock {\em J Vision}, 17(10):657--657.

\bibitem[Beurle and Matthews, 1956]{Beurle1956}
Beurle, R.~L. and Matthews, B. H.~C. (1956).
\newblock Properties of a mass of cells capable of regenerating pulses.
\newblock {\em Philosophical Transactions of the Royal Society of London.
  Series B, Biological Sciences}, 240(669):55--94.

\bibitem[Bezanson et~al., 2017]{bezanson2017julia}
Bezanson, J., Edelman, A., Karpinski, S., and Shah, V.~B. (2017).
\newblock Julia: A fresh approach to numerical computing.
\newblock {\em SIAM review}, 59(1):65--98.

\bibitem[Blakeslee et~al., 2016]{Blakeslee2016}
Blakeslee, B., Cope, D., and McCourt, M.~E. (2016).
\newblock The oriented difference of gaussians ({ODOG}) model of brightness
  perception: Overview and executable \emph{{M}athematica} notebooks.
\newblock {\em Behav. Res. Methods}, 48(1):306--312.

\bibitem[Bohi et~al., 2017]{Bohi2017}
Bohi, A., Prandi, D., Guis, V., Bouchara, F., and Gauthier, J.-P. (2017).
\newblock Fourier descriptors based on the structure of the human primary
  visual cortex with applications to object recognition.
\newblock {\em Journal of Mathematical Imaging and Vision}, 57(1):117--133.

\bibitem[Boscain et~al., 2018]{Boscain2018}
Boscain, U.~V., Chertovskih, R., Gauthier, J.-P., Prandi, D., and Remizov, A.
  (2018).
\newblock Highly corrupted image inpainting through hypoelliptic diffusion.
\newblock {\em Journal of Mathematical Imaging and Vision}, 60(8):1231--1245.

\bibitem[Bressan, 2001]{bressan2001explaining}
Bressan, P. (2001).
\newblock Explaining lightness illusions.
\newblock {\em Perception}, 30(9):1031--1046.

\bibitem[Brucke, 1865]{bruke}
Brucke, E. (1865).
\newblock uber erganzungs und contrasfarben.
\newblock {\em Wiener Sitzungsber}, 51.

\bibitem[Chariker et~al., 2016]{Chariker2016}
Chariker, L., Shapley, R., and Young, L.-S. (2016).
\newblock Orientation selectivity from very sparse lgn inputs in a
  comprehensive model of macaque v1 cortex.
\newblock {\em Journal of Neuroscience}, 36(49):12368--12384.

\bibitem[Citti and Sarti, 2006]{Citti2006}
Citti, G. and Sarti, A. (2006).
\newblock A cortical based model of perceptual completion in the
  roto-translation space.
\newblock {\em JMIV}, 24(3):307--326.

\bibitem[Cowan et~al., 2016]{Cowan2016}
Cowan, J.~D., Neuman, J., and van Drongelen, W. (2016).
\newblock Wilson--cowan equations for neocortical dynamics.
\newblock {\em The Journal of Mathematical Neuroscience}, 6(1):1.

\bibitem[Daugman, 1985]{Daugman1985a}
Daugman, J.~G. (1985).
\newblock {Uncertainty relation for resolution in space, spatial frequency, and
  orientation optimized by two-dimensional visual cortical filters.}
\newblock {\em J. Opt. Soc. Am. A}, 2(7):1160--1169.

\bibitem[DeValois and DeValois, 1990]{devalois1990spatial}
DeValois, R.~L. and DeValois, K.~K. (1990).
\newblock {\em Spatial vision}.
\newblock Oxford university press.

\bibitem[Duits et~al., 2007]{duits2007image}
Duits, R., Felsberg, M., Granlund, G., and ter Haar~Romeny, B. (2007).
\newblock Image analysis and reconstruction using a wavelet transform
  constructed from a reducible representation of the euclidean motion group.
\newblock {\em International Journal of Computer Vision}, 72(1):79--102.

\bibitem[Duits and Franken, 2010]{Duits2010}
Duits, R. and Franken, E. (2010).
\newblock Left-invariant parabolic evolutions on {$SE(2)$} and contour
  enhancement via invertible orientation scores. {Part I}: linear
  left-invariant diffusion equations on {$SE(2)$}.
\newblock {\em Quart. Appl. Math.}, 68(2):255--292.

\bibitem[Eagleman, 1959]{Eagleman2001}
Eagleman, D.~M. (1959).
\newblock Visual illusions and neurobiology.
\newblock {\em Nature reviews Neuroscience}, (2):920–926 (2001).

\bibitem[Ernst et~al., 2016]{Ernst2016}
Ernst, U.~A., Schiffer, A., Persike, M., and Meinhardt, G. (2016).
\newblock Contextual interactions in grating plaid configurations are explained
  by natural image statistics and neural modeling.
\newblock {\em Frontiers in systems neuroscience}, 10:78.

\bibitem[Fairhall et~al., 2001]{Fairhall2001}
Fairhall, A.~L., Lewen, G.~D., Bialek, W., and van Steveninck, R. R. d.~R.
  (2001).
\newblock Efficiency and ambiguity in an adaptive neural code.
\newblock {\em Nature}, 412(6849):787.

\bibitem[Faugeras, 2009]{Faugeras2009}
Faugeras, O. (2009).
\newblock {A constructive mean-field analysis of multi population neural
  networks with random synaptic weights and stochastic inputs}.
\newblock {\em Frontiers in Computational Neuroscience}, 3.

\bibitem[French, 2004]{French2004}
French, D. (2004).
\newblock Identification of a free energy functional in an integro-differential
  equation model for neuronal network activity.
\newblock {\em Applied Mathematics Letters}, 17(9):1047 -- 1051.

\bibitem[Geier and Hud{\'a}k, 2011]{geier2011changing}
Geier, J. and Hud{\'a}k, M. (2011).
\newblock Changing the chevreul illusion by a background luminance ramp:
  lateral inhibition fails at its traditional stronghold-a psychophysical
  refutation.
\newblock {\em PloS One}, 6(10):e26062.

\bibitem[Herzog et~al., 2003]{Herzog2003}
Herzog, M.~H., Ernst, U.~A., Etzold, A., and Eurich, C.~W. (2003).
\newblock Local interactions in neural networks explain global effects in
  gestalt processing and masking.
\newblock {\em Neural Computation}, 15(9):2091--2113.

\bibitem[Hong and Shevell, 2004]{hong2004brightness}
Hong, S.~W. and Shevell, S.~K. (2004).
\newblock Brightness contrast and assimilation from patterned inducing
  backgrounds.
\newblock {\em Vision Research}, 44(1):35--43.

\bibitem[Hubel and Wiesel, 1959]{HubelWiesel59}
Hubel, D. and Wiesel, T. (1959).
\newblock Receptive fields of single neurones in the cat's striate cortex.
\newblock {\em The Journal of physiology}, 148(3):574–591.

\bibitem[Kingdom, 2011]{Kingdom2011}
Kingdom, F.~A. (2011).
\newblock Lightness, brightness and transparency: A quarter century of new
  ideas, captivating demonstrations and unrelenting controversy.
\newblock {\em Vision Research}, 51(7):652--673.

\bibitem[Kitaoka, 2006]{kitaoka}
Kitaoka, A. (2006).
\newblock Adelson’s checker-shadow illusion-like gradation lightness
  illusion.
\newblock
  \url{http://www.psy.ritsumei.ac.jp/~akitaoka/gilchrist2006mytalke.html}.
\newblock Accessed: 2018-11-03.

\bibitem[Mante et~al., 2005]{Mante2005}
Mante, V., Frazor, R.~A., Bonin, V., Geisler, W.~S., and Carandini, M. (2005).
\newblock Independence of luminance and contrast in natural scenes and in the
  early visual system.
\newblock {\em Nature neuroscience}, 8(12):1690.

\bibitem[McCourt, 1982]{mccourt1982spatial}
McCourt, M.~E. (1982).
\newblock A spatial frequency dependent grating-induction effect.
\newblock {\em Vision Research}, 22(1):119--134.

\bibitem[Meister and Berry, 1999]{Meister1999}
Meister, M. and Berry, M.~J. (1999).
\newblock The neural code of the retina.
\newblock {\em Neuron}, 22(3):435--450.

\bibitem[Murray and Herrmann, 2013]{Murray2013}
Murray, M.~M. and Herrmann, C. (2013).
\newblock Illusory contours: a window onto the neurophysiology of constructing
  perception.
\newblock {\em Trends Cogn Sci.}, (17(9)):471--81.

\bibitem[Olshausen and Field, 2000]{Olshausen2000}
Olshausen, B.~A. and Field, D.~J. (2000).
\newblock Vision and the coding of natural images: The human brain may hold the
  secrets to the best image-compression algorithms.
\newblock {\em AmSci}, 88(3):238--245.

\bibitem[Otazu et~al., 2008]{Otazu2008}
Otazu, X., Vanrell, M., and Parraga, C.~A. (2008).
\newblock Multiresolution wavelet framework models brightness induction
  effects.
\newblock {\em Vis. Res.}, 48(5):733 -- 751.

\bibitem[Pattadkal et~al., 2018]{Pattadkal2018}
Pattadkal, J.~J., Mato, G., van Vreeswijk, C., Priebe, N.~J., and Hansel, D.
  (2018).
\newblock Emergent orientation selectivity from random networks in mouse visual
  cortex.
\newblock {\em Cell reports}, 24(8):2042--2050.

\bibitem[Petitot, 2017]{Petitot}
Petitot, J. (2017).
\newblock {\em Elements of Neurogeometry: Functional Architectures of Vision}.
\newblock Lecture Notes in Morphogenesis. Springer International Publishing.

\bibitem[Prandi and Gauthier, 2017]{Prandi2017}
Prandi, D. and Gauthier, J.-P. (2017).
\newblock {\em {A semidiscrete version of the Petitot model as a plausible
  model for anthropomorphic image reconstruction and pattern recognition}}.
\newblock SpringerBriefs in Mathematics. Springer International Publishing,
  Cham.

\bibitem[Purves et~al., 2008]{Purves2008}
Purves, D., Wojtach, W.~T., and Howe, C. (2008).
\newblock {V}isual illusions: {A}n {E}mpirical {E}xplanation.
\newblock {\em Scholarpedia}, 3(6):3706.
\newblock revision \#89112.

\bibitem[Ratliff, 1965]{ratliff1965mach}
Ratliff, F. (1965).
\newblock {\em Mach bands: quantitative studies on neural networks}.
\newblock Holden-Day, San Francisco London Amsterdam.

\bibitem[Sarti and Citti, 2015]{sarti2015constitution}
Sarti, A. and Citti, G. (2015).
\newblock The constitution of visual perceptual units in the functional
  architecture of {V1}.
\newblock {\em J. comput. neurosc.}, 38(2):285--300.

\bibitem[Scholl et~al., 2013]{Scholl2013}
Scholl, B., Tan, A.~Y., Corey, J., and Priebe, N.~J. (2013).
\newblock Emergence of orientation selectivity in the mammalian visual pathway.
\newblock {\em Journal of Neuroscience}, 33(26):10616--10624.

\bibitem[Shapiro and Todorovic, 2016]{shapiro2016oxford}
Shapiro, A.~G. and Todorovic, D. (2016).
\newblock {\em The Oxford compendium of visual illusions}.
\newblock Oxford University Press.

\bibitem[Smirnakis et~al., 1997]{Smirnakis1997}
Smirnakis, S.~M., Berry, M.~J., Warland, D.~K., Bialek, W., and Meister, M.
  (1997).
\newblock Adaptation of retinal processing to image contrast and spatial scale.
\newblock {\em Nature}, 386(6620):69.

\bibitem[Teich and Qian, 2010]{Teich2010}
Teich, A.~F. and Qian, N. (2010).
\newblock V1 orientation plasticity is explained by broadly tuned feedforward
  inputs and intracortical sharpening.
\newblock {\em Visual neuroscience}, 27(1-2):57--73.

\bibitem[Weintraub and Krantz, 1971]{Weintraub1971}
Weintraub, D.~J. and Krantz, D.~H. (1971).
\newblock {The Poggendorff illusion: amputations, rotations, and other
  perturbations.}
\newblock {\em Attent. Percept. Psycho.}, 10(4):257--264.

\bibitem[White, 1979]{white1979new}
White, M. (1979).
\newblock A new effect of pattern on perceived lightness.
\newblock {\em Perception}, 8(4):413--416.

\bibitem[Wilson, 2003]{Wilson2003}
Wilson, H.~R. (2003).
\newblock Computational evidence for a rivalry hierarchy in vision.
\newblock {\em Proceedings of the National Academy of Sciences},
  100(24):14499--14503.

\bibitem[Wilson, 2007]{Wilson2007}
Wilson, H.~R. (2007).
\newblock Minimal physiological conditions for binocular rivalry and rivalry
  memory.
\newblock {\em Vision Research}, 47(21):2741 -- 2750.

\bibitem[Wilson, 2017]{Wilson2017}
Wilson, H.~R. (2017).
\newblock Binocular contrast, stereopsis, and rivalry: Toward a dynamical
  synthesis.
\newblock {\em Vision Research}, 140:89--95.

\bibitem[Wilson and Cowan, 1972]{WilsonCowan1973}
Wilson, H.~R. and Cowan, J.~D. (1972).
\newblock Excitatory and inhibitory interactions in localized populations of
  model neurons.
\newblock {\em BioPhys. J.}, 12(1).

\bibitem[Wilson and Cowan, 1973]{Wilson1973bis}
Wilson, H.~R. and Cowan, J.~D. (1973).
\newblock A mathematical theory of the functional dynamics of cortical and
  thalamic nervous tissue.
\newblock {\em Kybernetik}, 13(2):55--80.

\end{thebibliography}

\end{document}